\documentclass[onecolumn]{article}
\usepackage[square]{natbib}
\usepackage{times}
\usepackage{amsmath}
\usepackage{amssymb,bbm}
\usepackage{amsthm}
\usepackage[hyperindex,breaklinks]{hyperref}

\usepackage{color}
\usepackage{algpseudocode}
\usepackage{algorithm}
\usepackage{enumitem}
\def\BState{\State\hskip-\ALG@thistlm}

\makeatletter  % Add this!!!!
 \renewcommand{\ALG@name}{Mechanism} %doesnt work
\makeatother   % Add this also

\newtheorem{theorem}{Theorem}[section]
\newtheorem{lemma}[theorem]{Lemma}
\newtheorem{assumption}[theorem]{Assumption}
\newtheorem{proposition}[theorem]{Proposition}
\newtheorem{definition}[theorem]{Definition}

\ifodd 0
\newcommand{\rev}[1]{#1}
\newcommand{\frev}[1]{{\color{magenta}#1}}%revise of the text
\newcommand{\com}[1]{\textbf{\color{red}(COMMENT: #1)}} %comment of the text
\newcommand{\clar}[1]{\textbf{\color{green}(NEED CLARIFICATION: #1)}}
\newcommand{\response}[1]{\textbf{\color{magenta}(RESPONSE: #1)}} %response to comment
\else
\newcommand{\rev}[1]{#1}
\newcommand{\frev}[1]{#1}
\newcommand{\com}[1]{}
\newcommand{\clar}[1]{}
\newcommand{\response}[1]{}
\fi
\newcommand{\RNum}[1]{\uppercase\expandafter{\romannumeral #1\relax}}

\newcommand{\squishlist}{
   \begin{list}{$\bullet$}
    { \setlength{\itemsep}{0pt}      \setlength{\parsep}{3pt}
      \setlength{\topsep}{3pt}       \setlength{\partopsep}{0pt}
      \setlength{\leftmargin}{1.5em} \setlength{\labelwidth}{1em}
      \setlength{\labelsep}{0.5em} } }
\newcommand{\squishend}{  \end{list}  }

\usepackage{geometry}
\begin{document}

\title{Learning to Incentivize: Eliciting Effort via Output Agreement}
%\coltauthor{Yang Liu}
%\editor{}
\author{Yang Liu and Yiling Chen \\ 
Harvard University, Cambridge MA, USA\\
\{yangl,yiling\}@seas.harvard.edu}
\maketitle

\begin{abstract}
\rev{
In crowdsourcing when there is a lack of verification for contributed answers, output agreement mechanisms are often used to incentivize participants to provide truthful answers when the correct answer is hold by the majority. In this paper, we focus on using output agreement mechanisms to elicit effort, in addition to eliciting truthful answers, from a population of workers. We consider a setting where workers have heterogeneous cost of effort exertion and examine the data requester's problem of deciding the reward level in output agreement for optimal elicitation. In particular, when the requester knows the cost distribution, we derive the optimal reward level for output agreement mechanisms. This is achieved by first characterizing Bayesian Nash equilibria of output agreement mechanisms for a given reward level. When the requester does not know the cost distribution, we develop sequential mechanisms that combine learning the cost distribution with incentivizing effort exertion to approximately determine the optimal reward level.} 
\end{abstract}

\section{Introduction}

Our ability to reach an unprecedentedly large number of people via the Internet has enabled crowdsourcing as a practical way for knowledge or information elicitation. For instance, crowdsourcing has been widely used for getting labels for training samples in machine learning. One salient characteristic of crowdsourcing is that a requester often cannot verify or evaluate the collected answers, because either the ground truth doesn't exist or is unavailable or it is too costly to be practical to verify the answers. This problem is called {\em information elicitation without verification} (IEWV)~\cite{Waggoner:14}. 

In the past decade, researchers have developed a class of economic mechanisms, collectively called the {\em peer prediction} mechanisms~\citet{Prelec:2004,MRZ:2005,jurca2006minimum,
jurca2009mechanisms,witkowski2012robust,witkowski2012peer,radanovic2013,Frongillo_aaai15}, for IEWV. The goal of most of these mechanisms is to design payment rules such that participants truthfully report their information at a game-theoretic equilibria. Each of these mechanisms makes some restriction on the information structure of the participants. Under the restriction, truthful elicitation is then achieved by rewarding a participant according to how his answer compares with those of his peers. Within this class, {\em output agreement mechanisms} are the simplest and they are often adopted in practice~\cite{von2004labeling}. In a basic output agreement mechanism, a participant receives a positive payment if his answer is the same as that of a random peer and zero payment otherwise. When the majority of the crowd hold the correct answer, output agreement mechanisms can truthfully elicit answers from the crowd at an equilibrium.    

Most of these works on peer prediction mechanisms, with the exception of \citeauthor{dasgupta2013crowdsourced} [\citeyear{dasgupta2013crowdsourced}] and \citeauthor{Witkowski_hcomp13} [\citeyear{Witkowski_hcomp13}], assume that answers of participants are exogenously generated, that is, participants are equipped with their private information. However, in many settings, participants can exert more effort to improve their information and hence the quality of their answers is endogenously determined. Recent experiments \cite{Yin_ijcai15,ho2015incentivizing} have also shown that the quality of answers can be influenced by the magnitude of contingent payment in settings where answers can be verified.  

In this paper, we study eliciting efforts as well as truthful answers in output agreement mechanisms. Taking the perspective of a requester, we ask the question of how to optimally set the payment level in output agreement mechanisms when the requester cares about both the accuracy of elicited answers and the total payment. 

Specifically, we focus on binary-answer questions and binary effort levels. We allow workers to have heterogeneous cost of exerting effort. Such a cost is randomly drawn from a distribution that is common knowledge to all participants. We consider two scenarios. In the first scenario, a static setting, the requester is assumed to know the cost distribution of the participants. Her objective is to set the payment level in output agreement mechanisms such that when a game-theoretic equilibrium is reached, her expected utility is maximized. In the second scenario, a dynamic setting, the data requester doesn't know the cost distribution of the participants but only knows an upper bound of the cost. Here, the requester incorporates eliciting and learning the cost distribution into incentivizing efforts in output agreement mechanisms when she repeatedly interacts with the set of participants over multiple tasks. The ultimate goal of the requester is to learn to set the optimal payment level in this sequential variant of output agreement mechanism for each interaction so that when participants reach a game-theoretic equilibrium of this dynamic game, the data requester minimizes her regret on expected utility over the sequence of tasks.  

We summarize our main contributions as follows:
\squishlist
\item Since the quality of answers is endogenously determined, a requester's utility depends on the behavior of participants. Optimizing the payment level requires an understanding of the participant's behavior. We characterize Bayesian Nash equilibria (BNE) for two output agreement mechanisms with any given level of payment and show that \rev{at equilibrium there is a unique threshold effort exertion strategy that returns each worker highest expected utility, when there is no collusion among workers}.
%\com{I think it's ok to not bother with the $P_L = 0.5$ case in the intro.} %As far as we are aware of, this is the first equilibrium characterization of output agreement mechanisms for effort elicitation when agents have heterogenous costs of effort exertion. 
\item For the static setting where the requester knows the cost distribution, when the cost distribution satisfies certain conditions, we show that the optimal payment level in the two output agreement mechanisms is a solution to a convex program and hence can be efficiently solved. 
\item For the dynamic setting where the requester doesn't know the cost distribution, we design a sequential mechanism that combines eliciting and learning the cost distribution with incentivizing effort exertion in a variant of output agreement mechanism. Our mechanism ensures that participants truthfully report their cost of effort exertion when asked, in addition to following the same strategy on effort exertion and answer reporting as that in the static setting for each task. We further prove performance guarantee of this mechanism in terms of the requester's regret on expected utility. 
\squishend

\subsection{Related work}\label{sec:related}

The literature on peer prediction mechanisms hasn't addressed costly effort until recently. \citeauthor{dasgupta2013crowdsourced} [\citeyear{dasgupta2013crowdsourced}] and \citeauthor{Witkowski_hcomp13} [\citeyear{Witkowski_hcomp13}] are the two papers that formally introduce costly effort into models of information elicitation without verification. \citeauthor{dasgupta2013crowdsourced} [\citeyear{dasgupta2013crowdsourced}] design a mechanism that incentivizes maximum effort followed by truthful reports of answers in an equilibrium that achieves maximum payoffs for participants. \citeauthor{Witkowski_hcomp13} [\citeyear{Witkowski_hcomp13}] focuses on simple output agreement mechanisms as this paper. They study the design of payment rules such that only participants whose quality is above a threshold participate and exert effort. Both \citeauthor{dasgupta2013crowdsourced} [\citeyear{dasgupta2013crowdsourced}] and \citeauthor{Witkowski_hcomp13} [\citeyear{Witkowski_hcomp13}] assume that the cost of effort exertion is fixed for all participants and is known to the mechanism designer. This paper studies effort elicitation in output agreement mechanisms but allow participants to have heterogeneous cost of effort exertion drawn from a common distribution. Moreover, we consider a setting where the mechanism designer doesn't know this cost distribution, which leads to an interesting question of learning to optimally incentivize effort exertion followed by truthful reports of answers in repeated interactions with a group of participants.  

%\citeauthor{cai2014optimum} [\citeyear{cai2014optimum}] also consider costly effort for eliciting data but for the purpose of statistical estimation. Similar to our setting, they allow participants to have heterogeneous cost of effort exertion, but they assume that such costs are fixed and known to the mechanism designer. 
\citeauthor{roth2012conducting} [\citeyear{roth2012conducting}] and \citeauthor{abernethy2015actively} [\citeyear{abernethy2015actively}] consider strategic data acquisition for estimating the mean and statistical learning in general respectively. Both works do not consider costly effort but participants may have stochastic and heterogeneous cost for revealing their data and need to be appropriately compensated. Moreover, these two works all assume that workers won't misreport their obtained answers. 
 
\frev{
\emph{Caveats:} With output agreement mechanisms, workers can achieve an uninformative equilibrium by colluding, which returns a higher utility for each worker. Our model and current results do not remove this caveat. For static scenario, it is promising to adopt the method introduced in \cite{dasgupta2013crowdsourced} to rule out such a case; nevertheless we conjecture that ruling out collusions in a dynamic setting with returning workers is much more challenging. This merits a future study.
}
\label{sec:related}

\section{Problem formulation}
\label{sec:pf}
\subsection{Our mechanisms}
%\com{I've changed $\mathcal N$ to $K$ to use all mathcal notations for sets. I think we don't use $\mathcal N$ anywhere else.}
A data requester has a set of tasks that she wants to obtain answers from a crowd $\mathcal C = \{1,...,K\}$ of $K \geq 2$ candidate workers. 
%We use $\mathcal C = \{1,...,K\}$ to denote this set of workers. 
In this paper, we consider binary-answer tasks, for example, identifying whether a picture of cells contains cancer cells, \rev{and denote the answer space of each task as $\{0, 1\}$. The requester assigns each task to $N$ randomly selected workers, with $N\geq 2$ being potentially much less than $K$.\footnote{We assume $N$ is fixed, though how to optimally choose $N$ could be an interesting future direction.} Such a redundant assignment strategy, when combined with some aggregation method (e.g. majority voting), has been found effective in obtaining accurate answers~\cite{sheng2008get,sig15}.} 

%following previous success in observing that redundant assignments (and then followed by label aggregation, e.g., majority voting) can greatly improve accuracy in obtained labels \cite{sheng2008get,sig15}. \rev{For simplicity the data requester will assign each task to a randomly selected subset of workers, with its size $N \geq 2$ being potentially much less than $\mathcal N$.\footnote{How to choose such $N$ optimally is an interesting future direction to extend this work. }}  \rev{After collecting outputs from assigned workers, the data requester cannot verify the correctness of an answer by herself. }

\rev{The requester cannot verify the correctness of contributed answers for a task, either because ground truth is not available or verification is too costly and defies the purpose of crowdsourcing. Thus, in addition to a base payment, each worker is rewarded with a contingent bonus that is determined by how his answer compares with those of other workers for completing a task. Specifically: } 
%Workers are rewarded with a base payment and a contingent bonus for completing a task: 
\begin{enumerate}
\item[1.] The requester assigns a task to \rev{a randomly selected subset $\mathcal U \subseteq \mathcal C$ of workers, where $|\mathcal U| = N$}. She announces a base payment $b > 0$ and a bonus \rev{$B > 0$}, as well as the criteria for receiving the bonus. The criteria of receiving the bonus is specified by an output agreement mechanism, which we will introduce shortly 
\item[2.] \rev{Each worker $i \in \mathcal U$ independently submits his answer $L_i \in \{0, 1\}$ to the requester.}
%Each worker will first make decision on whether to exert efforts or not. Then workers report their output for the assigned tasks.

\item[3.] After collecting the answers, the requester pays base payment $b$ to every worker who has submitted an answer and a bonus $B$ to those who met the specified criteria.  
\end{enumerate}

\rev{The criteria for receiving bonus $B$ is specified by an output agreement mechanism. \emph{Output agreement} is a term introduced by \citeauthor{von2008designing} [\citeyear{von2008designing}] to capture the idea of ``rewarding agreement'' in their image labeling game, the ESP game~\cite{von2004labeling}.  %Let the answer or label generated by worker $i$ for the given task as $L_i \in \{0,1\}$. 
We define two variants of output agreement mechanisms:}

%\com{Dropped the index $t$ from $L_i(t)$ for now because we didn't use it in the mechanism definition.}

%As we mentioned earlier, since there is no ground-truth label for validation purpose, we are going to adopt two output agreement mechanisms to define our rules for bonus. Denote the labels generated for each task $t$ by worker $i$ as $L_i(t) \in \{0,1\}$.

\vspace{-0.15in}
\paragraph{Peer output agreement (PA):} %The first mechanism we adopt is the classical peer output agreement mechanism: 
For each worker $i \in \mathcal U$, the data requester randomly selects a \emph{reference worker} $j\ne i$ and $j \in \mathcal U$. If $L_i = L_j$, worker $i$ receives bonus $B$. \rev{Note worker $j$'s reference worker could be different from $i$. }%\com{I guess we allow $j$ to have a reference worker that is not $i$, right? That's why we do not say that both workers receive the bonus.} 

\vspace{-0.15in}
\paragraph{Group output agreement (GA):} 
\rev{For each worker $i \in \mathcal U$, the data requester compares $L_i$ with the majority answer of the rest of the workers, $L_M$, where $L_M =1$ if $\frac{\sum_{j \in \mathcal U, j \neq i} L_j}{N-1} > 0.5$, $L_M =0$ if $\frac{\sum_{j \in \mathcal U, j \neq i} L_j}{N-1} < 0.5$ and $L_M = \{0,1\}$ if $\frac{\sum_{j \in \mathcal U, j \neq i} L_j}{N-1} = 0.5$. If $L_i \in L_M$, worker $i$ receives bonus $B$.}%\com{I think we need to take average for other agents' report. Please check whether I got these all correct.}\response{I need to make $L_M$ a set..later to simplify the analysis I assume a bonus will be given whenever there is tie, regardless of worker's outcome.}
\subsection{Agent models}
A worker can decide how much effort to exert to complete a task and the quality of his answer stochastically depends on his chosen effort level. Specifically, a worker can choose to either exert or not exert effort. If a worker exerts effort, then with probability $P_H\leq 1$ his answer is correct. If a worker does not exerts effort, with probability $P_L$, where $P_L <P_H$, he will provide the correct answer. We further assume $P_L \geq 0.5$, that is, when no effort is exerted the worker can at least do as well as random guess. This assumption is also used by \citeauthor{dasgupta2013crowdsourced} [\citeyear{dasgupta2013crowdsourced}] and \citeauthor{shah:sigmetrics13} [\citeyear{shah:nips13}, \citeyear{shah:sigmetrics13}]. For now, we assume $P_L$ and $P_H$ are the same for all workers.

Since workers can choose their effort level, the quality of an answer is endogenously determined. Let $e_i \in \{0, 1\}$ represents the chosen effort level of worker $i$, with $0$ corresponding to not exerting effort and $1$ corresponding to exerting effort.  
%Let $e_i \in \{0(\text{L}),1(\text{H})\}$ represents the chosen effort level of worker $i$. 
The accuracy of worker $i$ can be represented as 
$
p_i (e_i)= P_H  e_i + P_L  (1-e_i).
$

Workers have heterogeneous abilities, which are reflected by their cost of exerting effort. When worker $i$ doesn't exert effort, he incurs zero cost. A cost of $c_i \geq 0$ is incurred if agent $i$ chooses to exert effort on a task.  
%Exerting efforts necessarily incurs cost. To capture the difference in workers' expertise level, we consider the following effort-cost model. Suppose exerting efforts incurs a cost $c_i \geq 0$ for each worker $i$ on each task, with 
$c_i$ is randomly generated according to a distribution with pdf $f(c)$ and cdf $F(c)$ \rev{for each pair of (worker,~task)}. We further assume this distribution stays the same\footnote{Realization for each (worker, task) pair can be very different.} across all workers and all tasks, and it has a bounded support $[0,c_{\max}]$. Moreover we enforce the following assumption on $F(c)$:
\begin{assumption}
$F(c)$ is strictly concave on $c \in [0,c_{\max}]$.
\end{assumption}
This assumption is stating that the probability of having a larger cost $c_i$ is decreasing. Several common distributions, e.g. exponential and standard normal (positive side), satisfy this assumption. Throughout this paper, we assume $F(c), c \in [0,c_{\max}]$ is common knowledge among all workers.\footnote{In practice each worker can estimate such distribution based on their past experiences.} Nevertheless each realized cost $c_i$ is private information, that is each worker $i$ observes his own realized cost $c_i$, but not the one for others. In Section \ref{sec:static}, we assume the requester also has full knowledge of $F(\cdot)$, but we relax this assumption in Section \ref{sec:data}.

Given that the cost of not exerting effort is zero, the positive base payment $b$ ensures that every worker will provide an answer for a task assigned to him. We focus on understanding how to determine the bonus $B$ in output agreement mechanisms to better incentivize effort in this paper. The base payment $b$ doesn't enter our analysis directly but it allows us to not worry about workers' decisions on participation. %\com{I moved this discussion here.}   
\rev{
When reporting their answer to the data requester, workers can choose to report truthfully, or to mis-report. Denote this decision variable for each worker $i$ as $r_i \in \{0 , 1\}$, where $r_i=1$ represents worker $i$ truthfully reporting his answer, and $r_i=0$ represents worker $i$ mis-reporting (reverting the answer in our case). Then the accuracy of each worker $i$'s report is a function of $(e_i,r_i)$:
\begin{align*}
p_i(e_i,r_i) = p_i(e_i) r_i + (1-p_i(e_i))(1-r_i)~.
\end{align*}
}
%Notice we will assume workers will always participate when $b>0$ and no cost will be incurred when there is no effort exerted.  Past empirical results showed that the offered base payment only affects participation rate, but not the effort exertion. So adding the extra parameter, while will surely complicate the presentation, will not affect our analysis of the bonus mechanism. In this study we focus on incentivizing high efforts exertion with bonus, so it is not entirely unreasonable to make such a decoupling. 

%\begin{enumerate}
%\item[1.] Exponential distribution: $F(c) = 1-\lambda e^{-\lambda c}$.
%\item[2.] Uniform distribution: $F(c) = \frac{c}{c_{\max}}$.
%\item[3.] Standard normal distribution: \\$F(c) = \int_{\infty}^{c}\frac{1}{\sqrt{2\pi}}e^{-x^2/2}dx$, for $c \geq 0$.
%\end{enumerate}

%\com{assumption on common knowledge of c.}

%\subsection{Output agreement based bonus mechanism}

\rev{When each worker $j \in \mathcal U$ takes actions $(e_j, r_j)$, we denote the probability that worker $i \in \mathcal U$ receives bonus $B$ as $P_{i,B}( \{(e_j,r_j)\}_{j})$. In the PA mechanism, this quantity is \begin{align*}
&P_{i,B}( \{(e_j,r_j)\}_{j}) =  \frac{\sum_{j \neq i} P(L_i = L_j)}{N-1}%\\
%&=\frac{\sum_{j \neq i} p_i(e_i,r_i) p_j(e_j,r_j) + (1- p_i(e_i,r_i))(1-p_j(e_j,r_j))}{N-1}~.
\end{align*}
In the GA mechanism, it is   
$
P_{i,B}( \{(e_j,r_j)\}_{j}) = P(L_i = L_M).~%\label{prob:mv}
$
Then, the utility for worker $i$ is: 
\begin{align*}
u_i(\{(e_j,r_j)\}_{j}) = b - e_i c_i + B\cdot P_{i,B}( \{(e_j,r_j)\}_{j})~.%B\cdot P^B_i(\{e_j\}_{j \in \mathcal U}, b, B)~,
\end{align*}}

\subsection{Requester model}
\rev{
The data requester has utility function $U_D$, which in theory can be of various forms balancing accuracy of elicited answers and total budget spent. In this paper, we assume that the requester uses majority voting to aggregate elicited answers and has utility function  
$$
U_D(B) = P^c(N,B) - b~N - B~N_e(B),~
$$
where $P^c(N,B)$ is the probability that the majority answer is correct, and $N_e(B)$ is the number of workers who receive the bonus. Data requester's goal is to find a $B^*$ s.t.
\begin{align}
B^* \in \text{argmax}_{B\in \mathbbm R^+} P^c(N,B) - b~N - B~E[N_e(B)]~.\label{opb}
\end{align}
}

Notice both $P^c(N,B)$ and $E[N_e(B)]$ depend on workers' strategy towards effort exertion and answer reporting. The equilibrium analysis in the next section will help us define these quantities rigorously. The data requester is then hoping to choose a reward level that maximizes the expected utility at an equilibrium.

%\com{I changed $\mathbbm B$ to $\mathbbm R^+$. Is that ok? The original text is below.}
%
%Denote the data requesters' utility function as $U_D$. Theoretically speaking, $U_D$ can be of various forms balancing between accuracy of outputs and total budgets spent. We demonstrate with the following simple one: %Suppose the $U_D$ of interests consists of three parts:
%$
%U_D(B) = P^c(N,B) - b N - BN_e(B),~
%$
%where $P^c(N,B)$ is the probability that a majority voted outcome is correct (accuracy after label aggregation), and $N_e$ is the number of users who successfully collected the bonus. The data requesters' target is to choose a $B$ such that
%\begin{align*}
%B \in \text{argmax}_{B\in \mathcal B} P^c(N,B) - b N - BE[N_e(B)]~.
%\end{align*}

\section{Optimal bonus strategy with known cost distribution}\label{sec:static}

%\com{I think we want to emphasize that this section is for the case when the data requester knows the cost distribution and trying to solve the optimal bonus problem. I suggest we use a different section title and have a paragraph at the beginning to explain that characterizing the equilibrium is the means for us to find optimal bonus (i.e. the goal is to derive optimal bonus). Then maybe we want to change section 3.1 title to equilibrium characterization.}

%\com{I may be a good idea to mention that we only need to perform the analysis for one task.}

\rev{
In this section we set out to find the optimal bonus strategy when the data requester knows workers' cost distribution. Because the requester's utility depends on the behavior of workers, we first characterize symmetric Bayesian Nash Equilibria (BNE) for the two output agreement mechanisms for an arbitrary bonus level $B$. Then based on workers' equilibrium strategies, we show the optimal $B^*$ can be calculated efficiently for certain cost distributions. Note that due to the independence of tasks, this is a static setting and we only need to perform the analysis for a single task.}

\subsection{Equilibrium characterization}

For any given task, we have a Bayesian game among workers in $\mathcal U$. A worker's strategy in this game is a tuple $(e_i(c_i), r_i(e_i) )$ where $e_i(c_i): [0, c_{\text{max}}] \rightarrow \{0, 1\}$ specifies the effort level for worker $i$ when his realized cost is $c_i$ and $r_i(e_i): \{0, 1\} \rightarrow \{0, 1\}$ gives the reporting strategy for the chosen effort level, with $r_i(e_i) =1$ representing reporting truthfully and $r_i(e_i) = 0$ representing misreporting.  

We first argue that at any Bayesian Nash equilibrium (BNE) of the game, $e_i(c_i)$ must be a threshold function. That is, there is a threshold $c^*_i$ such that $e_i(c_i) =1$ for all $c_i \leq c^*_i$ and $e_i(c_i) =0$ for all $c_i > c^*_i$ at any BNE.  The reason is as follows: suppose at a BNE worker $i$ exerts effort with cost $c_i$. Since the other workers' outputs do not depend on $c_i$ (due to the independence of reporting across workers), worker $i$'s chance of getting a bonus will not change when he has a cost $c'_i < c_i$ and only obtains a higher expected utility by exerting effort. This allow us to focus on threshold strategies for effort exertion. We restrict our analysis to symmetric BNE where every worker has the same threshold for effort exertion, i.e. $c^*_i = c^*$. In the rest of the paper, we often use $(c^*, \cdot)$ to denote that a worker playing an effort exertion strategy with threshold $c^*$. In addition, we use $r_i \equiv 1$ to denote the reporting strategy that $r_i(1) = r_i(0) =1$, i.e. always reporting truthfully for either effort level. 

\paragraph{PA:} We have the following results for the PA mechanism. 
\begin{lemma}
The strategy profile $\{(c^*,r_i \equiv 1)\}_{i\in \mathcal U}$ is a symmetric BNE for the PA game if
\begin{align}
2(P_H-P_L)F(c^*)+2P_L-1 = c^*/((P_H-P_L)B)~.\label{eqn:ne}
\end{align}
\label{lemma:bne}
\end{lemma}
%
%\rev{Note when $r_i \equiv 1$, worker $i$ will truthfully report no matter she decides to exert effort or not (independent of $e_i$). We will also shorthand $(c^*,r_i \equiv 1)$ as $(c^*, 1)$ when there is no confusion.}
Denote $B_{\text{PA}}:= \frac{c_{\max}}{(2P_H-1)(P_H-P_L)}$, \frev{the minimum bonus level needed to induce full effort exertion}. With above lemma, we have the following equilibrium characterization.
%\begin{theorem}
%%We have the following set of results characterizing the equilibrium threshold $c^*$ and its corresponding $B$:
%%\begin{enumerate}
%%\item[(1)]
%
%When $P_L>0.5$, there always exists a unique threshold $c^*>0$ such that $(c^*, 1)$ is a symmetric BNE :
%\begin{itemize}
%\item[(1.1)] When $B \geq B_{\text{PA}}$, $c^*=c_{\max}$.
%\item[(1.2)] O.w. $c^*$ is the unique solution to Eqn. (\ref{eqn:ne}).
%\end{itemize}
%
% When $P_L = 0.5$, $(0,1)$ is a symmetric BNE $\forall B$. Besides, there exists at most one more threshold policy $c^*>0$ such that $(c^*, 1)$ is an equilibrium:
%\begin{itemize}
%\item[(2.1)] When $ B < \frac{1}{2f(0)(P_H-P_L)^2}$, there is no such $c^*>0$.
%\item[(2.2)] When $B \geq B_{\text{PA}}$, $c^*=c_{\max}$.
%\item[(2.3)] O.w., $c^*$ is the solution to 
%$
%B = \frac{c^*}{2F(c^*)(P_H-P_L)^2}~.
%$
%\end{itemize}
%%%\item[(2)]
%
%%\end{enumerate}
%\label{bne:po}
%\end{theorem}

\begin{theorem}
%We have the following set of results characterizing the equilibrium threshold $c^*$ and its corresponding $B$:
%\begin{enumerate}
%\item[(1)]
%
When $P_L>0.5$, there always exists a \frev{unique threshold $c^*>0$} such that $(c^*, 1)$ is a symmetric BNE for the PA game:
\begin{itemize}
\item When $B \geq B_{\text{PA}}$, $c^*=c_{\max}$.
\item O.w. $c^*$ is the unique solution to Eqn. (\ref{eqn:ne}).
\end{itemize}
\label{bne:po}
\end{theorem}
\frev{
This theorem implies that among all symmetric BNE, the effort exertion strategy is unique for a given bonus level $B$ when $P_L > 0.5$. 
When $P_L = 0.5$, we can prove similar results for the existence of $c^*>0$ such that $(c^*, 1)$ is a symmetric BNE. But this is not a unique effort exertion strategy. In fact, we have a set of trivial symmetric BNE for all $B$: $(0, \cdot)$, that is no one exerting effort combined with any reporting strategy. \rev{Nonethelss this trivial equilibrium returns strictly less expected utility for each worker. }}%Unfortunately under this case, nobody exerting effort, i.e, $(0,1)$, is another symmetric BNE, $\forall B$. 
 
% Besides, there exists at most one more threshold policy $c^*>0$ such that $(c^*, 1)$ is an equilibrium:
%\begin{itemize}
%\item[(2.1)] When $ B < \frac{1}{2f(0)(P_H-P_L)^2}$, there is no such $c^*>0$.
%\item[(2.2)] When $B \geq B_{\text{PA}}$, $c^*=c_{\max}$.
%\item[(2.3)] O.w., $c^*$ is the solution to 
%$
%B = \frac{c^*}{2F(c^*)(P_H-P_L)^2}~.
%$
%\end{itemize}
%%\item[(2)]

%\end{enumerate}

%\com{Does it make sense to state the results for $P_L > 0.5$ first since it's the one that we'll use later?}\response{I've changed the order.}

%\com{Might worth to briefly explain what this theorem means in English. For example, uniqueness blah blah blah ... Am I right that it is only unique when we restrict to symmetric equilibria with threshold strategies? So, there may still be symmetric equilibria where workers do not use a threshold strategy? If that's right, let's be careful with the wordings.}
%\response{There is no other symmetric equilibria as I was trying to argue earlier that restricting to threshold policy is w/o loss of generality. Let me try to move the argument here then. }

\rev{
%Note when $f(0)=0$ we take $\frac{1}{2f(0)(P_H-P_L)^2}$ as $+\infty$. T
%The proof for effort exertion is largely algebraic; the uniqueness of the threshold policy $c^*>0$ is mainly due to the fact that a concave function and a linear function intersect with each other at most twice. For truth-telling, the reasoning is similar with previous peer prediction works: with $P_H \geq P_L > 0.5$\footnote{When the accuracy degenerates to 0.5, there is no difference between truth-telling and mis-reporting.}, that is when workers are more likely to be correct, they will \emph{race for accuracy}. 
We would like to note that always mis-reporting ($r_i \equiv 0$) combined with the same threshold $c^*$ for effort exertion as in Theorem \ref{bne:po} is also a symmetric BNE when $P_L > 0.5$. This equilibrium gives workers the same utility as the equilibrium in Theorem \ref{bne:po}. This phenomenal has also been observed by \citeauthor{dasgupta2013crowdsourced} [\citeyear{dasgupta2013crowdsourced}] and \citeauthor{Witkowski_hcomp13} [\citeyear{Witkowski_hcomp13}]. \citeauthor{dasgupta2013crowdsourced} [\citeyear{dasgupta2013crowdsourced}] argue that always mis-reporting is risky, and workers may prefer breaking the tie towards always truthful reporting. %But these are the only two symmetric BNE (with effort exertion taking a threshold policy) we have -- no other mixed strategy between truthful report and mis-report is an equilibrium.
}

\paragraph{GA:} For GA, directly calculating the probability term for matching a majority voting is not easy; but if we adopt a Chernoff type approximation for it, and suppose such approximation is common knowledge\footnote{This is not entirely unreasonable as in practice this Chernoff type bounds are often used to estimate such majority voting probability term.}, we can prove similar results. %Also the above solution provides a lower bound on the bonus level (for fixed threshold) or an upper bound on the effort level (for fixed bonus) for the original one.

Similar to Lemma~\ref{lemma:bne}, we can show that the strategy profile $\{(c^*,r_i \equiv 1)\}_{i\in \mathcal U}$ is a symmetric BNE for the GA game if
\begin{align}
1-2[(\alpha-1)F(c^*)+&1]^{N-1}= c^*/(B(P_H-P_L))~.\label{equ:eq:ga}%\\
%&~.\nonumber
\end{align}
where $\alpha:= e^{-2(P_H-P_L)^2}$. 

Denote $B_{\text{GA}}:= \frac{c_{\max}}{(1-2\alpha^{N-1})(P_H-P_L)}$ we have:
%thus we seek a relaxation. 
%Then if each worker replaces $P(\frac{\sum_{j \in [k]} I^H_j + \sum_{j \in \mathcal U_{-i}-[k]} I^L_k}{N-1} \geq 0.5)$ (where $I^H_j, I^L_j$ are the indicator variable under High and Low effort respectively) with a Chernoff type approximation $1 - e^{-2(P_H-P_L)^2k}$, 
\begin{theorem}
 When $P_L>0.5$, there always exists a unique threshold $c^*>0$ such that $\{(c^*,r_i \equiv 1)\}_{i\in \mathcal U}$ is a symmetric BNE for the GA game:
\begin{itemize}
\item When $B \geq B_{\text{GA}}$, $c^*=c_{\max}$.
\item O.w., $c^*$ is the unique solution to Eqn. (\ref{equ:eq:ga}).
\end{itemize}
% When $P_L = 0.5$, $(0,1)$ is a symmetric BNE $\forall B$. Besides, there exists at most one more threshold policy $c^*>0$ such that $(c^*, 1)$ is an equilibrium:
%\begin{itemize}
%\item[(2.1)] When $ B < \frac{1}{2(N-1)(\alpha^{N-2}-\alpha^{N-1})f(0)(P_H-P_L)}$, there is no such $c^*>0$.
%\item[(2.2)] When $B \geq B_{\text{GA}}$, $c^*=c_{\max}$.
%\item[(2.3)] O.w., $c^*$ is the unique solution to Eqn. (\ref{equ:eq:ga}).
%\end{itemize}
\label{eqn:mv}
\end{theorem} 

%\com{Maybe also switch the two results.}

Moreover we can show that the reward level and the total expected payment if lower in GA than in PA for eliciting the same level of efforts. 
%\subsection{Peer output agreement v.s. group output agreement}
%We now discuss the difference between the two types of bonus mechanisms.
\begin{lemma}
Denote the \rev{smallest} bonus level corresponding to an arbitrary equilibrium threshold $c^*>0$ for PA and GA as $B_{\text{PA}}(c^*)$ and $B_{\text{GA}}(c^*)$ respectively. Then
$
B_{\text{PA}}(c^*) > B_{\text{GA}}(c^*),~
$
\frev{when $N$ is sufficiently large \rev{(e.g., $N \geq \frac{-\log (1-P_H)}{2(P_L-0.5)^2}+1$)}}. Furthermore, the total payment in GA is lower than that in PA. \label{comp} 
\end{lemma}
%\com{If it's easy to state what sufficiently large means, we can state it out. Otherwise, it's ok as it is.}

This result also implies that adopting GA will lead to a higher requester utility.
%Similar result has also been observed in previous output agreement works \cite{}. \com{missing a ref.}

%\com{Should we say $B_{PA}(c^*)$ and $B_{GA}(c^*)$ are the smallest bonus levels to achieve the threshold $c^*$? When $c^* = c_max$, the bonus level is not unique.}\response{corrected.}

%\com{Just want to double check. This result not only shows that the bonus level for PA is higher than that for GA to induce the same threshold $c^*$ but also demonstrate that at the equilibrium, the total bonus collected is higher in PA than in GA. The total bonus collected not only depends on the level of the bonus but also depends on the probability of collecting a bonus.}\response{thats right! the total bonus is $NB\cdot P(B,c^*)$ where $ P(B,c^*)$ is the probability of a worker successfully collecting a bonus. GA achieves a smaller $NB\cdot P(B,c^*)$. }

\paragraph{Heterogeneity of $P_L$ and $P_H$.}

So far we have assumed that $P_L$ and $P_H$ are the same for all workers. If workers have heterogeneous accuracy $\{P^i_L,P^i_H\}_i$ that are generated from some distribution with mean $P_L, P_H$, we can show that the above results hold in a similar way, with more involved arguments. %Each worker is aware of such stochasticity and the mean, while they do not observe their own labeling quality, nor others. 

 %Each worker only knows the mean not the specific realization?
%\com{Calculate total cost. $B \cdot P$.}

\subsection{Optimal solution for data requester}

%With above equilibrium we can now calculate the optimal bonus strategy from the data requester's perspective. Solving the optimization problem for a general distribution is hard. But we show for certain category of distributions, it is convex and can be solved efficiently. 

%With the equilibrium analysis we now calculate $P^c(N,B)$ and $E[N_e(B)]$. 
Now consider the optimization problem stated in Eqn. (\ref{opb}) for the requester's perspective. For each $B>0$, denote $\{(c,r_i \equiv 1)\}_{i\in \mathcal U}$ as the corresponding strategy profile at equilibrium. $P^c(N,B)$ can then be calculated based on $c, F(c)$ (controlling how much effort can be induced), and $P_L,P_H$. Same can be done for $E[N_e(B)]$. Denote the optimization problem in (\ref{opb}) with above calculation as $(\texttt{PB})$. \rev{Directly investigating the two objective functions may be hard. We seek to relax the objectives. First of for PA, we will be omitting the $2P_L-1$ term as when $P_L$ is only slightly larger than 0.5, this quantity is close to 0. Also for both PA and GA, we again use the Chernoff type of approximation for calculating $P^c_B$. } We further introduce three conditions: (i) $f(c)$ is twice differentiable and $\partial^2 f(c)/\partial^2 c \geq 0$. (ii) $cF(c)$ is convex on $c \in [0,c_{\max}]$. (iii) $G(c):=1-[(\alpha-1)F(c^*)+1]^{N-1}$ satisfies that $\partial^3 G(c)/\partial^3 c$ exists and being non-negative.
\begin{lemma}
If (i) and (ii) hold, the objective function of $(\texttt{PB})$ is concave if we adopt PA. When (ii) and (iii) hold, the objective function of $(\texttt{PB})$ is concave if we adopt GA.
\label{convex}
\end{lemma}
\rev{
For example, exponential distribution (exp($\lambda$)) for $c_{\max} \leq 2/\lambda$ satisfies (i)\&(ii) for PA; and exp($\lambda$) for $c_{\max} = \frac{-\log \alpha}{\lambda}$ satisfies (ii)\&(iii) for GA.} It is worth to note above results hold for a wide range of other $U_D(\cdot)$s: for instance the ones with a linear combination of $P^c(N,B)$ and $E[N_e(B)]$.

\section{Learning the optimal bonus strategy}
\label{sec:data}

 In this section we propose a sequential mechanism for learning the optimal bonus strategy, when the requester has no prior knowledge of the cost distribution but only knows $c_{\max}$. This assumption can be further relaxed by assuming knowing an upper bound of $c_{\max}$ instead of  knowing $c_{\max}$ precisely.  Also similar as last section, $P_L, P_H$ are known. In reality these two quantities can be estimated through a learning procedure by repeated sampling and output matching as shown in \cite{sig15}, via setting bonus level $B:=0$ and $B:=B_{\text{PA(GA)}}$\footnote{Calculating $B_{\text{PA(GA)}}$ only requires the knowledge of $c_{\max}$, not $F(\cdot)$.} respectively (to induce effort level corresponding to $P_L,P_H$). In this work we focus on learning the cost functions, which is a more challenging task when the workers are strategic. We are in a dynamic setting where the requester sequentially ask workers to complete a set of task. In our mechanisms, requesters can ask workers to report their costs of effort exertion for a task and based on their reports decide on the bonus level for the current task and for future tasks in a output-agreement-style mechanism.   
 \vspace{-0.15in}
\paragraph{(P1):} We start our discussions with a simpler case. %Workers are myopic. 
When asked to report their cost, workers maximize their collected utility from a set of data elicitation tasks and are not aware of the potential influence of their reports on calculating optimal bonus levels for any future tasks.  
%Suppose the workers are only aware of the data elicitation stages, and their goal is to maximize the utilities they can obtain from all these elicitation tasks (workers are myopic). %This is a simplification of our problem. As in fact workers may start to realize the elicited data will be utilized to play the optimal bonus strategy against them, in which case they will have stronger incentives to deviate from truthfully reporting their data. 
The data requester's goal is to elicit cost data to estimate cost distribution and then the optimal bonus level $\tilde{B}^*$, such that when $\tilde{B}^*$ is applied to a \emph{newly arrived} task we can bound 
 $
 |U_D(\tilde{B}^*)-U_D(B^*)|.
 $
 where $B^*$ is the optimal bonus level if the cost distribution is known.   \vspace{-0.15in}
\paragraph{(P2):} We then consider the case when workers are forward looking and are aware of that their reported cost on a task will be utilized to calculate optimal bonus strategy for future tasks. We form a sequential learning setting, where we separate the stages for task assignment into two types: one for data elicitation, which we also refer as \emph{exploration}, and the other for utility maximization, which we refer as \emph{exploitation}. The data requester's objective in this case is to minimize the \emph{regret} defined as follows:
%\begin{align}
%R(T) = \sum_{t=1}^T E|U_D(B(t))-U_D(B^*_{\text{MV}})|~,
%\end{align}
\begin{align}
R(T) = \sum_{t=1}^T E|U_D(\{B_i(t)\}_{i \in \mathcal U})-U_D(B^*_{\text{GA}})|~,
\end{align}
where $B^*_{\text{GA}}$ is the optimal bonus level for GA when cost distribution is know,\footnote{Since GA is more cost efficient, we define the regret w.r.t. the optimal utility that can be obtained via GA.} and $\{B_i(t)\}_{i \in \mathcal U}$ is the bonus bundle offered at time $t$.  \rev{
Note $U_D(\cdot)$ is mechanism dependent: both $P^c(N,B)$ and $N_e(B)$ depends on not only the bonus level, but also the equilibrium behavior in a particular mechanism.
%Note $U_D(\cdot)$ is mechanism dependent: $U_D(\{B_i(t)\}_{i \in \mathcal U})$ potentially can differ from $U_D(B^*_{\text{GA}})$ in both $P^c(N,B)$ and $N_e(B)$.
}

%%While our assumption for this simple case sounds a bit risky, it is not entirely impossible to implement the corresponding solution. Suppose the number of workers is large enough such that we can afford to separate the workers into two disjoint sets $\mathcal U = \mathcal U_e \cup \mathcal U_b, ~\mathcal U_e \cap \mathcal U_b = \emptyset$. We then dedicate $\mathcal U_e$ for data elicitation and the other one $\mathcal U_b$ for labeling tasks. Such separation will help de-couple workers' incentives in reporting the cost data. Nevertheless, %, and more coupled argument and proof will be needed. 

For simplicity of presentation, throughout this section we consider $P_L > 0.5$: this is to remove the ambiguity introduced in by the trivial equilibrium $c^* = 0$. Also we assume with the same expected utility, workers will favor truthful reporting $r_i \equiv 1$. %, that is at equilibrium, workers will truthfully report. 
%\end{itemize}
%\com{I thought (3) was also assumed in the previous section, since $c^*$ depends on $P_L$ and $P_H$. Maybe I'm missing something here.}\response{Yes thats right -- (3) was also assumed in last section. Since we are releasing the assumption that data requester knowns the cost distribution, I though I should justify why the $P_L,P_H$ are still being known.}

\begin{algorithm}[!h]
\caption{ (\texttt{M\_Crowd})}\label{m:simple}
\begin{algorithmic}
\State For each step $t$:
\begin{itemize}
\item[1.]  Assign the task, and workers then report costs. Denote the reports as $(\tilde{c}_1(t),...,\tilde{c}_N(t))$. This is a voluntary procedure. A worker $i$ can choose to not report his cost, in which case, the requester sets $\tilde{c}_i(t):=c_{\max}$.
\item[2.] Data requester randomly selects a threshold $c^*(t)$ \emph{uniformly} from the support $[0,c_{\max}]$, such that only the workers who reported $\tilde{c}_i(t) \leq c^*(t)$ will be considered for bonus following PA; others will be given a bonus according to a probability that is independent of workers' report (see Remarks for details). %denied from the bonus stage, \rev{though they will still receive a base payment.}
%\com{What happens to other workers? Do they still participate and get a base payment? Or are they denied participation?}
\item[3.] The requester estimates a bonus level $\tilde{B}_i(t)$ for each worker $i$ that corresponds to the threshold level $c^*(t)$ under PA, using \emph{only} the data collected from user $j \neq i$, and from \emph{all} previous stages. \rev{This is done via estimating $F(\cdot)$ first and then plugging it into Eqn.(\ref{eqn:ne}). }%\com{I don't quite understand how the estimation is done with data collected from other users and from all previous stages. Do we basically use all previous data (excluding $i$'s) to estimate $F(c)$ and then solve optimal $B$ according to (3)? If we are mixing across tasks, it may worth to justify it at some point (not in the mechanism description though) as it crucially depends on the assumption that the cost distribution is the same for all tasks.}\response{I added a paragraph in the remark. }
Then the requester adds a positive perturbation $\delta(t)$ that is time dependent to $\tilde{B}_i(t)$:
$
B_i(t) := \tilde{B}_i(t) +\delta(t).
$
\item[4.] The data requester will then announce the bonus bundle $[B_1(t),...,B_N(t)]$.
\end{itemize}
\end{algorithmic}
\end{algorithm}

\subsection{(\texttt{M\_Crowd}) for \textbf{(P1)}}

Suppose the data requester allocates $T$ tasks to elicit the cost data sequentially, and exactly one of them is assigned to the workers at each time step $t=1,2,\cdots,T$.  For simplicity of analysis we fix the set of $N$ workers we will be assigning tasks to. Denote worker $i$'s realized cost for the $t$-th task as $c_i(t)$. We propose mechanism (\texttt{M\_Crowd}):

\paragraph{Remarks:} 1. When a worker, say worker $i$, reports higher than the selected threshold, his probability of receiving a bonus will be calculated using the following experiment, which is independent of his output: suppose out of $N$ workers, there are $N(t)$ of them reported lower than $c^*(t)$. Then we will "simulate" $N$ workers' reports with the following coin-toss procedure: toss $N(t)$ $P_H$-coin and $N-N(t)$ $P_L$-coin. Assign a $P_L$-coin toss to worker $i$, and select a reference answer from the rest of the tosses, and compare their results. If there is a match, worker $i$ will receive a bonus. Simply put, the probability for receiving a bonus can be calculated as the matching probability in the above experiment. 2. Since we have characterized the equilibrium equation for PA with a clean and simple form, this set of equilibriums is good for eliciting workers' data. 3. After estimating the bonus level for each worker $\tilde{B}_i(t)$, the data requester will add a positive perturbation term $\delta(t)$ to each of them. This is mainly to remove the bias introduced by (i) imperfect estimation due to finite number of collected samples, and (ii) the (possible) mis-reports from workers. Such term will become clear later in the stated results. \rev{4. The fact that we can use collected cost data to estimate $B$ depends crucially on the assumption that the cost distribution is the same for all tasks.} %This assumption can be further relaxed by separating the learning processes for different types of tasks.}

%\begin{remark}
%We have a set of remarks before proceeding.
%\begin{itemize}
%\item[1.] Notice we adopted peer agreement mechanism to elicit workers' cost samples. The reason is since we have precisely characterized the equilibrium for peer agreement, this set of equilibriums is better for eliciting users' data. 
%\item[2.] When we estimate the bonus level for each worker $B_i(t)$, we add positive perturbations $\delta(t)$ to each of them. This is mainly to remove the bias introduced by (i) imperfect estimation due to finite number of collected samples, and (ii) the (possible) mis-reports from workers. Such term will become self-clear later through our proofs. %\com{should mention the order of this term is on $O(\sqrt{\frac{\log t}{t}})$.}
%\end{itemize}
%\end{remark}

%For the rest of the section, we will first characterize the incentive compatible and individual rational reporting strategy for each worker, and then we characterize the regret in calculating optimal bonus strategy with this set of reported data. 

\subsection{Equilibrium analysis for (\texttt{M\_Crowd})}

%\com{I really like how you try to provide intuition on the analysis in this section! But I think the paragraphs before the theorem and the proof sketch have a lot of overlap. They can be shortened or merged.}
\frev{
We present the main results for characterizing workers' cost data reporting strategies at an equilibrium. Because of the independence of tasks, the effort exertion on each stage is essentially a static game. While for cost reporting, even though we are in a dynamic game setting, we again adopt BNE as our general solution concept. It may sound more intuitive to use Perfect Bayesian equilibrium (PBE) to define our solution in dynamic setting, but we argue BNE and PBE does not make conceptual difference in our case. Note that workers' decision on effort exertion is not directly observable by others, and the only signals one worker can use to update their belief towards others' effort exertion actions are the offered bonus levels. However due to the stochastic nature of the calculated $\tilde{B}_i(t)$, any realization is on the equilibrium path with positive probability (though could be arbitrarily small). Simply put, in our case there is no off-equilibrium path information set. 

Let $u^t_i(\cdot)$ denotes worker $i$'s utility at time $t$. We will adopt $\epsilon$ approximate BNE as our exact solution concept, which is defined as follows:
\begin{definition}
A set of reporting strategy $\{\mathbf{\tilde{c}_i}:= \{\tilde{c}_i(t)\}_{t=1,...,T}\}_{i \in \mathcal C}$ is $\epsilon$-BNE if for any $i$, $\forall \mathbf{\tilde{c}'_i} \neq \mathbf{\tilde{c}_i}$ we have
\begin{align*}
&\sum_{t=1}^T E[\max_{e_i,r_i}u^t_i(\mathbf{\tilde{c}_i}, \mathbf{\tilde{c}_{-i}})]/T \geq \sum_{t=1}^T E[\max_{e_i,r_i} u^t_i(\mathbf{\tilde{c}'_i}, \mathbf{\tilde{c}_{-i}})]/T - \epsilon~.
\end{align*}
\end{definition}
Here we explicitly denote the expected utility for each worker as a function of $\{\mathbf{\tilde{c}_i}\}_{i \in \mathcal C}$. Note this is rather a short-hand notation, as $u^t_i$s also depend on the effort exertion and reporting strategies. The $\max_{e_i,r_i}u^t_i(\mathbf{\tilde{c}_i}, \mathbf{\tilde{c}_{-i}})$ term allows worker $i$ to optimize his effort exertion and reporting procedure based on their cost reporting.
%Notice following our mechanism, workers' strategy on effort exertion does not form a dynamic game, as the assigned tasks are independent with each other. The cost data report stages do form a dynamic game; but at each stage worker's reporting will only affect the future but not the past.  
}

\begin{theorem}
With  (\texttt{M\_Crowd}), set $\delta(t) := O(\sqrt{\log t/t})$, let $\gamma>0$ being arbitrarily small, there exists a \frev{$O(\frac{(\log T)^2}{T})$-BNE} for each worker $i$ with reporting $\tilde{c}_i(t)$ at time $t$ such that
\begin{align*}
\max\{c_i(t) -\epsilon_1(t),0 \} \leq \tilde{c}_i(t) \leq \min\{c_i(t)+\epsilon_2(t),c_{\max}\}~,%\epsilon_2(t)~,
\end{align*}
where $0 \leq \epsilon_1(t)= o(\sqrt{\log t/t}),~ 0\leq \epsilon_2(t) = o(1/t^{2-\gamma})$.
\label{thm:report}
\end{theorem}
%\com{approximate PBE?}
%\rev{In particular, towards establishing the above we show:
%\begin{lemma}
%At each step $t$, after reporting their cost $\{\tilde{c}_i(t)\}$ according to the strategy in Theorem~\ref{thm:report}, workers who reported less than or equal to $c^*(t)$ will exert effort and report truthfully, while those who reported higher than $c^*(t)$ will not exert effort, and report truthfully at a BNE for (\texttt{M\_Crowd}).
%\end{lemma}
%}

The effort exertion game at each step looks alike the static game introduced in Section \ref{sec:static} with the following difference: instead of workers who have cost $c_i(t) \leq c^*$ will exert effort, now it is the workers who reported $\tilde{c}_i(t) \leq c^*$ will exert effort. This is mainly due to the addition of the perturbation term to the estimated bonus level. Meanwhile the mechanism excludes workers who reported higher than the threshold from exerting effort by offering bonus with a probability that is independent of worker's output. Nevertheless, we can bound the fraction of workers whose actions are different for the above two games. We provide intuitions for the proof.

%the game for effort exertion among workers will change slightly. This is primarily due to two reasons: (1) we enforce a thresholding step for bonus. (2) Noisy calculation of bonus level $B$. But also due to (1) and the fact . At the meantime we will eliminate the ones with higher cost to stay out the bonus stage. Therefore regarding the utility for each worker, the new game looks very similar to the static game i

%We discuss over-reporting and under-reporting separately. 

 %We separate the discussions to over-reporting and under-reporting respectively.

\vspace{-0.15in}
\paragraph{Over-reporting}

%We prove the following results.
%\begin{lemma}
%With  (\texttt{M-Crowd}), set $\delta(t) := O(\sqrt{\frac{\log t}{t}})$, suppose worker $i$ over-reports by $\sigma_i(t)$ at time $t$, that is $\tilde{c}_i(t)=\min\{c_i(t)+\sigma_i(t),c_{\max}\}$, we will have reporting with
%\begin{align*}
%\sigma_i(t) \leq O(\frac{1}{t^{2-\gamma}})~,
% \end{align*} 
% where $\gamma>0$ is an arbitrarily small quantity, is incentive compatible for each worker $i$. \label{over-report}
%\end{lemma}
%At any time $t$, once a worker deviates by \rev{over-reporting $c_i(t)+\sigma_i(t)$}, this will cause $O(\frac{\sigma_i(t)}{t'}), t' \geq t$ extra bias in estimating the cost distribution for any further step $t'$.\com{Is $\sigma_i(t)$ the reported cost of agent $i$ for task $t$? I think it wasn't defined. Also, might it be mixed up with $\sigma(t)$ which is the perturbation?} We map such bias in distribution to the discrepancy in estimating bonus $B_j, j \neq i$. %Thanks to Lemma \ref{lemma:f},  
Over-reporting by worker $i$ will mislead the data requester into believing that finishing the tasks costs more than it actually is, so this can lead to a higher estimation of $\tilde{B}_j(t'), j \neq i, t' \geq t$. %Notice by design of our algorithm, such an over-reporting will not change each worker $i$'s own bonus level directly. But rather 
This will induce more effort from other workers, which will in turn increase the utility for worker $i$. We bound the extra efforts exerted from other workers $j \neq i$, where we will be utilizing the second step of  (\texttt{M\_Crowd}) that the one reported higher than the threshold will be excluded from exerting effort (since worker's probability of winning a bonus will be independent of her report), and the decoupling step (Step 3) where each worker's bonus level will only be calculated over data collected from others.  On the other hand, over-reporting will decrease the chance of receiving bonus (excluded from effort exertion). %Also due to under-reports from the other workers, worker $i$ may benefit from over-reporting to be excluded from effort exertion -- in which case the number of $P_H$ coins will be larger, so will worker $i$'s chance of matching a toss.
%However due to the reporting mechanism, these additional workers with $c$ falls between $[B(t), \tilde{B}(t)]$ ($\tilde{B}(t)$ denotes the estimated bonus level with such over-reporting) will not be allowed to participate, except for with probability $\frac{1}{t^2}$ where $B(t)$ can be arbitrarily smaller than the target bonus (Chernoff type event).\com{I didn't follow the previous sentence. From the mechanism, I know that if $c > c^*(t)$, the agent is excluded from the bonus. Did you establish a relationship between $c^*(t)$ and $B(t)$, and $\tilde{B}(t)$ so that most agents with cost falling into this range has a cost that's higher than $c^*(t)$?} %By Lemma \ref{lip:g} 
%Then we can bound the fraction of workers who changed their decisions from non effort exertion to effort exertion, due to the increase of bonus level. Such gain can be upper bounded by
%\begin{align*}
%O(\sum_{t'=t}^{\infty} \frac{\sigma_i(t)}{t'} \cdot \frac{1}{(t')^2})\leq O(\frac{\sigma_i(t)}{t^{2-\gamma}}\cdot \sum_{t'=t} \frac{1}{(t')^{1+\gamma}}) = O(\frac{\sigma_i(t)}{t^{2-\gamma}})~.
%\end{align*} 
%On the other hand, the loss due to over-reporting (by missing a bonus chance) can be lower bounded by $O(\sigma^2_i(t))$, with which we conclude $\sigma_i(t) = o(\frac{1}{t^{2-\gamma}})$ is incentive compatible.

\vspace{-0.15in}
\paragraph{Under-reporting}

%\begin{lemma}
%%The (greedy)-optimal under-report level is upper bounded by the solution for 
%%\begin{align*}
%%\delta^*(t) = c_{\max}[\frac{2}{t^2}\cdot \bar{B} + (1-\frac{1}{t^2})(\delta_1(t)+\epsilon(t))]~,
%%\end{align*}
%%%where $\epsilon(t) = \frac{1}{\sqrt{(N-1)t/\log t}}$, and $\bar{B}$ is an upper bound for bonus.
%%and $\bar{B}$ is an upper bound for bonus.
%With  (\texttt{M-Crowd}), set $\delta(t) := O(\sqrt{\frac{\log t}{t}})$, suppose worker $i$ under-reports by $\sigma_i(t)$ at time $t$, that is $\tilde{c}_i(t)=\max\{c_i(t)-\sigma_i(t),0\}$, we will have reporting with
%\begin{align*}
%\sigma_i(t) \leq O(\sqrt{\log t} \cdot t^{-1/2})~,
%\end{align*}
%is incentive compatible for each worker $i$.
%\label{under-report}
%\end{lemma}

When a worker under-reports, he will gain by having a higher bonus in expectation --  this is due to (i) the fact that the threshold is randomly determined, and (ii) we added positive perturbation to estimated bonus level. The loss is due to the fact that with under-reporting, with positive probability, exerting effort costs more than the threshold cost. The regulation for under-reporting mainly comes from the thresholding step of (\texttt{M\_Crowd}).

\com{now we have a longer version of technical report, we could afford to take out some parts?}

\subsection{Performance of (\texttt{M\_Crowd})}

With this set of collected data, we bound the performance loss in offering optimal bonus level $B$ for an incoming task (or task $T+1$). Suppose we adopt GA, where the optimal bonus level with known cost distribution is given by $B^*_{\text{GA}}$, and the estimated optimal solution is given by $\tilde{B}^*_{\text{GA}}$. We will have the following lemma: (similar results hold for PA)
\begin{lemma}
With probability being at least $1-\eta$, 
\begin{align*}
|U_D(\tilde{B}^*_{\text{GA}})-U_D(B^*_{\text{GA}})| = o( \sqrt{\frac{\log 2/\eta}{2NT}} + \sqrt{\frac{\log T}{T}})~.
\end{align*}
\label{perf:1}
\end{lemma}
When we chose $\eta = O(1/T^2)$, the above regret term is roughly on the order of $\sqrt{\log T/T}$~.

%\com{$U_D(B)$ was defined for a single task. How should I think about $U_D(\tilde{B}^*)$ and $U_D(B^*)$ here? Are they the utility for task $T+1$ or cumulative utility for previous $T$ tasks, or average utility? Need clarification, I think.}

%Suppose $|F(\tilde{c}^*)-F(c^*)| \geq \delta_1$. We will have $|\tilde{c^*}-c^*| \geq \delta_2$. 

%But clearly $c^*, \tilde{F(c^*)}$ will give a utility function bounded within $O(\epsilon)$, so this $\tilde{c}$ should also be bounded. \com{need to complete this argument with more precise derivation.}

\subsection{(\texttt{RM\_Crowd}) for \textbf{(P2)}}
%\label{sec:learn}

%Now we release the assumption that workers are un-aware of being exploited. 
We propose a (\texttt{RM\_Crowd}) for \textbf{(P2)}:

\begin{algorithm}
\caption{ (\texttt{RM\_Crowd})}\label{m:simple}
\begin{algorithmic}
\State Specify a constant $0<z<1$, and initialize $t=1$. Define $p(t):=\min\{1,\frac{\log T}{t^{1-z}}\}$.
\State At time $t$, assign the task; workers then report costs.
\State Toss a $p(t)$-coin.  
\State When HEAD, algorithm enters \emph{exploration} phase:
\begin{itemize}
\item Follow same steps as in (\texttt{M\_Crowd}).
%\item Toss a $\frac{1}{t^{z_2}}$=coin. If HEAD, the elicited data will be used for future calculation.
%\item First the data requester requires every worker to report their cost for the current job. Denote the reports as $(\tilde{c}_1(t),...,\tilde{c}_N(t))$. Notice it is possible $\tilde{c}_i(t) \neq c_i(t)$, that is workers can report arbitrarily. Moreover, this is a voluntary procedure, that is, a worker $i$ can choose to not report her cost. In this case we will simply set $\tilde{c}_i(t):=c_{\max}$.
%\item After collecting these reports, the data requester will randomly select a threshold $c^*(t)$ \emph{uniformly} random from the region $[0,c_{\max}]$, such that only the workers who reported $\tilde{c}_i(t) \leq c^*(t)$ will receive bonus (conditional on the fact his output matches the other reference peer). Others she will be eliminated from the bonus stage, regardless of her output.
%\item Estimate a bonus level $\tilde{B}_i(t)$ for EACH worker $i$ that corresponds to the threshold level $c^*(t)$ under \emph{peer agreement}, using ONLY the data collected from user $j \neq i$, and from all \emph{exploration} stages. Then the data requester will add a positive perturbation term $\delta(t)$ that is time dependent to $\tilde{B}_i(t)$:
%\begin{align*}
%B_i(t) := \tilde{B}_i(t) +\delta(t)~.
%\end{align*}
%\item The data requester will then announce the bonus bundle $[B_1(t),...,B_N(t)]$.
\end{itemize}
\State When TAIL, algorithm enters $\emph{exploitation}$ phase,
\begin{itemize}
\item[1.] estimate the optimal $B_i(t)$ and its corresponding threshold $c^*_i(t)$ for each worker $i$ with GA, using the cost data collected \emph{only} from the \emph{exploration} phases. Only workers who reported $c_i(t) \leq c^*_i(t)$ will be given bonus according to GA; others receive bonus with a probability that is independent of her report.
\item[2.] Follow rest steps in (\texttt{M\_Crowd}).
\end{itemize} 
\end{algorithmic}
\end{algorithm}

\paragraph{Remarks:} 1. The dependence on $T$ is to simplify the presentation and our algorithm design. This can be easily extended to a $T$-independent one. 2. At exploitation phases we assume there exists a solver that can find the optimal solution with a noisy estimation of $F(\cdot)$. In practice search heuristics can help achieve the goal. 3. We adopted different bonus mechanisms for different phases.  \rev{When we calculate the bonus level according to a particular mechanism (PA or GA), we will also adopt it for evaluating workers' answers. } 4. When using GA, the independent probability for giving out bonus when a report is higher than the threshold will be adjusted to a probability of matching a majority voting of the experiment we presented for (\texttt{M\_Crowd}). %(4) According to our mechanism, $\tilde{U}_D(\cdot)$ differs from $U_D(\cdot)$ in the number of effort exertions, and the probability of collecting a bonus. For example, instead of depending exactly on the threshold policy, $\tilde{P}^c(N,B)$ will be depending on how many workers actually exerted effort.% } %(2) We adopted different bonus mechanisms for different phases.  As we have shown group output agreement is more cost efficient, we stick with it for exploitation phases, which is by design more prominent. %(3) Since we offer different bonus to different workers, we need to refine the definition of learning regret $R(T)$ as follows:
%\com{Is using different mechanisms in these two phrases necessary or just preferred? Seems to be the latter. Also, it might be better to explicitly state that PA is then used in M\_Crowd and then GA is then used in RM\_Crowd. As it is stated now, we only say that the bonus level is calculated using the $B^*$ for PA or GA, but doesn't mention which mechanism is used for evaluating workers' answers.}
%\com{Do we need to justify why we use different mechanisms?}

%
%\begin{remark}
%We have a set of remarks before proceeding.
%\begin{enumerate}
%\item[1.] The dependence on $T$ is to simplify the presentation and our algorithm design. This can be easily removed, that is our algorithm can be easily extended to a $T$-independent one. 
%\item[2.] We adopted different scoring mechanisms for different phases.  As we have also shown group output agreement is more cost efficient, we stick with it for exploitation phases, which is by design more prominent.
%\item[3.] Since we offer different bonus to different workers, we need to refine the definition of learning regret $R(T)$ by a bit:
%\begin{align}
%R(T) = \sum_{t=1}^T E|U_D(\{B_i(t)\}_{i \in \mathcal U})-U_D(B^*_{\text{MV}})|~.
%\end{align}
%\end{enumerate}
%\end{remark}

%\subsubsection{Performance analysis}

%We analyze the incentive compatible and individually rational reporting strategies.
\begin{theorem}
With  (\texttt{RM\_Crowd}), set $\delta(t) := O(z/t^{z/2})$, let $z > 1/3$ and $\gamma>0$ being arbitrarily small, there exists a \frev{$O(\frac{z^2}{T^z})$-BNE} for each worker $i$ with reporting $\tilde{c}_i(t)$ at time $t$ such that
\begin{align*}
\max\{c_i(t) -\epsilon_1(t),0 \} \leq \tilde{c}_i(t) \leq \min\{c_i(t)+\epsilon_2(t),c_{\max}\}~,
\end{align*}
where $0 \leq \epsilon_1(t) = o(z/t^{z/2}),~ 0 \leq \epsilon_2(t) = o(1/t^{3z-1-\gamma})$.
% is incentive compatible. And such reporting is also individually rational for worker $i$ whose cost $c_i(t)$ satisfies:
%\begin{align*}
% \epsilon_1(t) \leq c_i(t) \leq c_{\max}-\epsilon_2(t)~.
% \end{align*} 
\label{learn2}
\end{theorem}
%The proof is very similar to the one we presented earlier. The subtle difference lies in the fact when a worker chooses to deviate, there is additional incentive due to the exploitation period, among which the number of available samples (for estimation purpose) grows slowly. Also 

%\com{add intuitions for that?}

We have similar observations for the effort exertion game in (\texttt{RM\_Crowd}) as we made for (\texttt{M\_Crowd}). Further we prove the following regret results:%When $z \rightarrow 1$ (full exploration), we can show the two bounds approach the ones we derived in last section. 
%\end{remark}

%\subsubsection{Characterizing the regrets}
%In this section we characterize the regret $R(T)$. 

\begin{lemma}
%$R(T)$ can be bounded as:
$
R(T) \leq O(T^z\log T + T^{1-z/2})~.
$
\label{perf:2}
\end{lemma}
Order-wise, the best $z$ is when $1-z/2 = z \Rightarrow z = \frac{2}{3}$, which leads to a bound on the order of $O(T^{2/3}\log T)$.

%\com{need to bound the actual error that the number of exploration phases is less than $O(T^z\log T)$. Using Hoeffding bound.}

%\input{discuss}

\section{Conclusion}\label{sec:conclude}
In this paper we focus on using output agreement mechanisms to elicit effort, in addition to eliciting truthful answers, from crowd-workers when there is no verification of their outputs. %We first quantify workers' incentives for effort exertion with two output agreement based bonus mechanisms. Different from previous literature, 
Workers' cost for exerting efforts are stochastic and heterogeneous. We characterize the symmetric BNE for workers' effort exertion and reporting strategies for a given bonus level, and show data requester's optimal bonus strategy at equilibrium is a solution to a convex program for certain cost distribution. Then a learning procedure is introduced to help the requester learn the optimal bonus level via eliciting cost data from strategic workers. We bound the mechanism's performance loss w.r.t. offering the best bonus bundle, compared to the case when workers' cost distribution is known a priori.

%An immediate extension of this work is to consider more sophisticated incentive mechanisms for deciding on giving out a reward, rather than restricting to the output agreement mechanisms.
%\frev{
%\section*{Acknowledgement}
%This paper is funded by XXX. 
%}
%

\bibliographystyle{named}
\bibliography{myref,library}
%\newpage
%\onecolumn
%\onehalfspacing

\section*{\LARGE Appendices}

\section{Proof for Lemma \ref{lemma:bne}}
\begin{proof}
We denote worker $i$'s expected utility\footnote{Throughout the proof we will interchange the wording between utility and bonus for workers' payoff.} for each task as $u_i(e_i,r_i; \{(e_j,r_j)_{j \neq i}\})$ when effort exertion and reporting strategy $\{(e_j,r_j)\}$ has been adopted by the crowd. And we shorthand it as $u_i(e_i, r_i)$. Also we denote by $u_i(e_i, r_i;c) $ worker $i$'s utility when a threshold policy $c$ for effort exertion is adopted.

First we show at equilibrium workers will not deviate from truthfully reporting $r_i=1$. Suppose all other workers $j \neq i$ are following $(c,1)$. We consider the following two cases. When $e_i=0$, the expected bonus worker $i$ can have by truthfully reporting is
 $\frac{B}{N-1}\sum_{j \neq i} P_L E(2p_j-1)$. By deviating we have the above term become $\frac{B}{N-1}\sum_{j \neq i} (1-P_L) E(2p_j-1)$. Since $E(2p_j-1) > 0$ under equilibrium strategy, we know by deviating worker's utility will decrease. Similarly we can show when $e_i = 1$, by mis-reporting, worker's utility decreases from  $\frac{B}{N-1}\sum_{j \neq i} P_H E(2p_j-1)$ to $\frac{B}{N-1}\sum_{j \neq i} (1-P_H) E(2p_j-1)$. For more discussions please refer to Section \ref{discuss:ne}.

Let workers adopt a threshold policy $c$, and truthfully report their labeling outcome $r_i=1$. Consider the difference in worker $i$'s expected bonus between exerting effort and not:
\begin{align}
&u_i(e_i = 1, r_i=1;c) - u_i(e_i = 0, r_i=1;c) \nonumber \\
&= -c_i + \frac{B}{N-1}\sum_{j \neq i} E[P_H p_j + (1-P_H)(1-p_j)] \nonumber \\
&~~~~~~- \frac{B}{N-1}\sum_{j \neq i}E [P_L p_j + (1-P_L)(1-p_j)] \nonumber \\
&= -c_i + \frac{B}{N-1}\sum_{j \neq i} (P_H-P_L)E(2p_j-1)~.\label{diff:bonus}
\end{align}
Consider the summation in Eqn. (\ref{diff:bonus}).
\begin{align*}
\sum_{j \neq i} E(2p_j-1)&=
2\sum_{k=0}^{N-1} {N-1 \choose k} F^{k}(c)(1-F(c))^{N-1-k} [kP_H+(N-1-k)P_L]-(N-1)\\
&= 2\sum_{k=0}^{N-1} {N-1 \choose k} F^{k}(c)(1-F(c))^{N-1-k}(P_H-P_L)k + 2(N-1)P_L-(N-1)~.
\end{align*}
Consider the sum of combinatorial terms in above equation:
\begin{align*}
&~~~~\sum_{k=0}^{N-1} {N-1 \choose k} F^{k}(c)(1-F(c))^{N-1-k} \cdot k \\
&=\sum_{k=1}^{N-1}\frac{(N-1)!}{(N-1-k)!\, (k-1)!}F^{k}(c)(1-F(c))^{N-1-k}\\
&=(N-1)F(c) \sum_{k=1}^{N-1}\frac{(N-2)!}{((N-2)-(k-1))!\,(k-1)!}F^{k-1}(c)(1-F(c))^{(N-2)-(k-1)}\\
&=(N-1)F(c) \sum_{k=0}^{N-2}\frac{(N-2)!}{((N-2)-k)!\,k!}F^{k}(c)(1-F(c))^{(N-2)-k}\\
&=(N-1)F(c)~.
\end{align*}
Then set $c_i=c$ and $u_i(e_i = 1, r_i=1;c) - u_i(e_i = 0, r_i=1;c) = 0$, i.e., when there is no difference between exerting effort and not,  we have
\begin{align*}
2(P_H-P_L)F(c)+2P_L-1 = \frac{c}{B(P_H-P_L)}~.
\end{align*}
With this it is easy to see when $c_i > c$, $ -c_i + \frac{B}{N-1}\sum_{j \neq i} (P_H-P_L)E(2p_j-1) < 0$, that is not exerting effort is a better move. While on the hand when $c_i < c$, $ -c_i + \frac{B}{N-1}\sum_{j \neq i} (P_H-P_L)E(2p_j-1) > 0$, worker should exert effort to maximize utility.

\end{proof}

\section{Proof for Theorem \ref{bne:po}, with extension to $P_L=0.5$ }

\begin{proof}
First consider the case $P_L>0.5$. When $c^* = 0$, we have
\begin{align*}
2(P_H-P_L)F(c^*)+2P_L-1 = 2P_L-1 > 0 = \frac{c^*}{B(P_H-P_L)}~.
\end{align*}
Since LHS is an increasing, strictly concave in $c$, and RHS is linear in $c$, we know the LHS and RHS can only intersects once on $\{c > 0\}$. Now we discuss in two cases. 
\begin{itemize}
\item[(1.1)]
First when 
$$
2(P_H-P_L)F(c_{\max})+2P_L-1 > \frac{c_{\max}}{B(P_H-P_L)},~
$$
we have $\forall c > 0$ (by concavity, as a combination of $(0, c_{\max})$: $c = \frac{c}{c_{\max}}\cdot c_{\max}+(1-\frac{c}{c_{\max}})\cdot 0$)
\begin{align*}
&2(P_H-P_L)F(c) +2P_L-1 \\
&\geq \frac{c}{c_{\max}}(2(P_H-P_L)F(c_{\max})+2P_L-1) + (1-\frac{c}{c_{\max}})(2P_L-1) \\
&>\frac{c}{c_{\max}} \frac{c_{\max}}{B(P_H-P_L)} + 0\\
&=\frac{c}{B(P_H-P_L)}~.
\end{align*}
This is implying that for all effort level $c$, it is better to exert effort than to not, in which case the only equilibrium is $c^* = c_{\max}$ which corresponds to full effort exertion. 
\item[(1.2)] For the second case, when
$$
2(P_H-P_L)F(c_{\max})+2P_L-1 \leq \frac{c_{\max}}{B(P_H-P_L)}~,
$$
there will be odd number of crossings between LHS and RHS of Eqn.(\ref{eqn:ne}) on $[0,c_{\max}]$, so there must exist only one of them (since there are two at most) corresponding to the solution of the equilibrium equation.
\end{itemize}
\end{proof}

Now consider $P_L = 0.5$. We first rigorously state our results.
\begin{lemma}
 When $P_L = 0.5$, $(0,1)$ is a symmetric BNE $\forall B$. Besides, there exists at most one more threshold policy $c^*>0$ such that $(c^*, 1)$ is an equilibrium:
\begin{itemize}
\item[(2.1)] When $ B < \frac{1}{2f(0)(P_H-P_L)^2}$, there is no such $c^*>0$.
\item[(2.2)] When $B \geq B_{\text{PA}}$, $c^*=c_{\max}$.
\item[(2.3)] O.w., $c^*$ is the solution to 
$
B = \frac{c^*}{2F(c^*)(P_H-P_L)^2}~.
$
\end{itemize}
\end{lemma}

\begin{proof}
When $c^*=0$, $F(c^*) = 0$, so no worker will exert efforts, and the LHS of Eqn.(\ref{eqn:ne}) reduces to 0, which matches the RHS. Moreover as a strictly concave function intersects with a linear function at most twice, there exists at most one more intersection point (equilibrium point). We again discuss in cases:
\begin{itemize}
\item[(2.1)] When the following holds,
$$
\frac{1}{B(P_H-P_L)^2} > 2f(0),~\text{or}~ B <\frac{1}{2f(0)(P_H-P_L)^2}
$$
we will have for $c > 0$,
$$
\frac{c}{B(P_H-P_L)} > 2c(P_H-P_L)f(0) \geq 2(P_H-P_L)F(c),
$$
where the last inequality is due to the concavity of $F(\cdot)$. So for any $c>0$, we have for Eqn. (\ref{eqn:ne}) $\text{LHS} < \text{RHS}$, i.e., no effort exertion is the only equilibrium, or equivalently $c^* \equiv 0$.

\item[(2.2)] When $B = B_{\text{PA}}$, we will have 
$$
2(P_H-P_L)F(c_{\max})  = \frac{c_{\max}}{(P_H-P_L)B}~.
$$
By strict concavity of $F(\cdot)$, we know for all $0<c<c_{\max}$ (as similarly reasoned for the case with $P_L > 0.5$), 
$$
2(P_H-P_L)F(c) > \frac{c}{B(P_H-P_L)}.
$$
Therefore we will also know by setting $B$ larger we will have for all $c>0$
$$
2(P_H-P_L)F(c) > \frac{c}{B(P_H-P_L)}.
$$
Hence effort exerting is always a better move, regardless of $c$.
\item[(2.3)] For the case in the middle, we know there exists (and guaranteed to exist) a unique intersection point, which solution can be obtained by simply setting 
$$
B(c^*) = \frac{c^*}{2F(c^*)(P_H-P_L)^2},
$$
and solve for $c^* > 0$. 
\end{itemize}
\end{proof}

Now we prove the following claim we made in the paper that the trivial equilibrium returns less expected utility than a non-trivial one $c^* > 0$. We demonstrate this under (PA). Similar results can be established for (GA).
\begin{lemma}
Under (PA), when there exists a non-trivial equilibrium $c^*>0$, it returns higher expected bonus than adopting the trivial one $c^*=0$.
\end{lemma}

\begin{proof}
We prove this by discussing two cases. When $c_i > c^*$, exerting effort returns higher expected bonus compared to no effort exertion, when every worker else is exerting efforts according to the threshold policy $c^*$:
\begin{align*}
&u_i(e_i = 1, r_i=1;c^*) > u_i(e_i = 0, r_i=1;c^*) \\
&= (P_L (F(c^*)P_H + (1-F(c^*))\cdot 0.5) \\
&~~~~~~+ (1-P_L) [1-(F(c^*)P_H + (1-F(c^*))\cdot 0.5)])\cdot B\\
& = \frac{B}{2}~,
\end{align*}
which is exactly the expected bonus a worker can have under the trivial equilibrium. Similarly when  $c_i< c^*$, we have
\begin{align*}
u_i(e_i = 0, r_i=1;c^*) = \frac{B}{2}~,
\end{align*}
with which we finish our proof.
\end{proof}

\section{Proofs for discussions on equilibrium strategies}\label{discuss:ne}

We provide proofs for some claims we made in the discussions of other equilibrium strategies. As a remainder, we summarize the results here:
\begin{itemize}
\item Mis-report also induces an equilibrium, which returns each worker the same utility (with the same effort exertion threshold).
\item There does not exist effort-dependent equilibrium.
%\item Collusion and cheap signal.
\end{itemize}
~\\
~\\
\emph{Truthful report}

\begin{lemma}
Always mis-reporting ($r_i \equiv 0$) combined with the same threshold $c^*$ for effort exertion as in Theorem \ref{bne:po} is also a symmetric BNE when $P_L > 0.5$.\label{lemma:0.5}
\end{lemma}

%We prove for all the non-degenerate threshold hold policy $c^*>0$ on equilibrium, truthful reporting is an equilibrium. 
\begin{proof}
Throughout this section we will shorthand the following notations:
\begin{align*}
p_i := p_i(e_i),~ \tilde{p}_i := p_i(e_i, r_i=0)~.
\end{align*}
%\begin{lemma}
%When $c^*>0$, truth-telling is an equilibrium.\label{truthtelling}
%\end{lemma}
%\begin{proof}
%When $c^*=0$, no one is exerting effort a.s., so $p_i = P_L = 0.5$ (note when $P_L > 0.5$, $c^* > 0$). Therefore the outputs from workers are statistically identical to random guesses. In this case being non-truth telling will not change each user's utility.
%
%For $c^*>0$. Consider worker $i$, and suppose every other worker is truthfully reporting. And similarly we can write down each worker $i$'s utility function as:
%\begin{align*}
%u_i(e_i;c)  &=   b-e_ic_i+ \frac{B}{N-1} \sum_{j \neq i}E[ p_i(e_i,r_i) p_j + (1-p_i(e_i,r_i))(1-p_j)]\\
%&=b-e_ic_i+BE[p_i(e_i,r_i)]\sum_{j \neq i} E(2p_j-1)+B\sum_{j \neq i} E(1-p_j)~.
%\end{align*}
%Notice with $j \neq i$ truthfully reporting, and the threshold policy $c^*$ being positive, we must have
%\begin{align*}
%E[\sum_{j \neq i} (2p_j-1)] > E[\sum_{j \neq i} (2P_L-1)] > 0~.
%\end{align*}
%We know by setting $r_i=1$ worker $i$ will achieve its maximum utility, i.e., by truthfully reporting worker will improve their expected utility from non-truthful ones.
%\end{proof}

We first show mis-reporting (reverting the answer in our binary labeling case) with the same effort level is also an equilibrium. 
When all other workers $j \neq i$ are mis-reporting, we will have
\begin{align*}
u_i(e_i,r_i=0;c) &= b-e_ic_i+ \frac{B}{N-1} \sum_{j \neq i} [\tilde{p}_i \tilde{p}_j + (1-\tilde{p}_i)(1-\tilde{p}_j)]\\
%&=b-e_ic_i+BE[\tilde{p}_i]E[\sum_{j \neq i} (2\tilde{p}_j-1)]+BE[\sum_{j \neq i} (1-\tilde{p}_j)]\\
&=b-e_ic_i+B\tilde{p}_i \sum_{j \neq i} (1-2p_j)+B \sum_{j \neq i} p_j~.
%\\
%&=b-e_ic_i+BE[\tilde{p}_i]E[\sum_{j \neq i} (1-2p_j)+BE[\sum_{j \neq i} p_j ]~.
\end{align*}
Since $1-2p_j < 0$, minimizing $\tilde{p}_i$,  i.e., by setting $r_i=0$, will maximize worker $i$'s utility. Also 
\begin{align*}
u_i(e_i,r_i=0;c) &= b-e_ic_i+ \frac{B}{N-1} \sum_{j \neq i} [\tilde{p}_i \tilde{p}_j + (1-\tilde{p}_i)(1-\tilde{p}_j)]\\
&= b-e_ic_i+ \frac{B}{N-1} \sum_{j \neq i} [p_i p_j + (1-p_i)(1-p_j)] \\
&= u_i(e_i,r_i=1;c)~,
\end{align*}
i.e., the two equilibrium returns the same expected bonus.
\end{proof}

\frev{
Under the binary signal case, in fact the above mis-reporting strategy is equivalent with permutation strategy as studied in \cite{2016arXiv160303151S}. It is known when signals are categorical \cite{2016arXiv160303151S} (which is our case), permutation strategy of an equilibrium remains as an equilibrium, and returns the same utility at equilibrium.
}

Also no other mixed strategy between truthful report and mis-report with mixing probability $0<\delta<1$ will be an equilibrium. Denote labeling accuracy for such mixed strategy for each worker $i$ as $\tilde{p}_i(\delta)$. Then observe that worker $i$'s utility can be written as follows:
\begin{align*}
b-e_ic_i+BE[\tilde{p}_i(\delta)]E[\sum_{j \neq i} (2\tilde{p}_j(\delta)-1)]+BE[\sum_{j \neq i} (1-\tilde{p}(\delta)_j) ]~.
\end{align*}
Depending on different $\delta$, $E[\sum_{j \neq i} (2\tilde{p}_j-1)]$ is either positive or negative. Correspondingly worker $i$ will have incentives to deviate to $\delta=0$ or $\delta=1$. 

%Even worse we notice the mis-report equilibrium returns same utility as the one for truthfully reporting. This is a phenomenal also observed in other two papers \cite{dasgupta2013crowdsourced,Witkowski_hcomp13}. As also noted in \cite{dasgupta2013crowdsourced}, being completely mis-reporting is risky, and the workers may prefer breaking the tie towards truthful reporting. 
It is also worth noting that this issue can be resolved by assuming a known prior on the label. We can then generate a random outcome based on the prior (by tossing a coin for example) and compare each worker's outcome with this random bit with certain probability. This will reduce the utility for untruthful reporting. To see this, suppose there is a certain prior on the labels or outcomes for each task. Denote the prior probability for getting a $H$ as $p^0_H$, and suppose we have this extra piece of information $p^0_H>0.5$. Then with probability $0<\epsilon<1$, instead of comparing worker $i$'s answer with a randomly selected peer, we randomly toss a coin with probability $p^0_H$, and compare the worker's answer to the outcome of this toss. This fraction of utility is then given by
\begin{align*}
\epsilon(p^0_H p_i(e_i,r_i) + (1-p^0_H)(1-p_i(e_i,r_i))) = \epsilon (p_i(e_i,r_i)(2p^0_H-1)+1-p^0_H)~,
\end{align*}
from which we observe truthful reporting will return a higher utility, as (i) we will have a higher $p_i(e_i,r_i)$ and (ii) $2p^0_H-1>0$.
~\\
~\\
~\\
\emph{Effort-dependent reporting}
\begin{lemma}
There does not exists effort-dependent reporting at equilibrium.
\end{lemma}
\begin{proof}
As similarly argued in the proof for Lemma \ref{lemma:0.5}, suppose there is an effort-dependent reporting equilibrium $r_i(e_i)$. Then again depending on the sign of $E[\sum_{j \neq i} (2p_j-1)]$ being either positive or negative, when $P_H>P_L>0.5$, worker $i$ will choose to either truthfully report on both cases or mis-report to match other's outputs. 

When $E[\sum_{j \neq i} (2p_j-1)] = 0$, worker $i$ has no incentive to exert effort -- so there cannot be an effort-dependent equilibrium.
% if $e_i \equiv 0$ and $P_L > 0.5$ (so $p_i \equiv P_L$), it is impossible to have $E[\sum_{j \neq i} (2\tilde{p}_j-1)] = 0$ for any symmetric equilibrium. 
\end{proof}

\section{Heterogeneous $P_L,P_H$  }

\begin{lemma}
The claims in Theorem \ref{bne:po} and \ref{eqn:mv} hold when $\{P^i_H,P^i_L\}_{i \in \mathcal C}$ are generated according to certain distribution with mean $P_H, P_L$. 
\end{lemma}
\begin{proof}
It is quite clear that with (PA) the previous results hold. The reason is under peer agreement, the utility function is linear in each worker's expertise level (consider truthful reporting):
\begin{align*}
&E[u_i(e_i = 1, r_i=1;c) - u_i(e_i = 0, r_i=1;c)] \\
&= -c_i +  \frac{B}{N-1}E[\sum_{j \neq i} (P^i_H-P^i_L)(2p_j-1)]\\
&= -c_i +  \frac{B}{N-1}E[P^i_H-P^i_L] E[2p_j-1] \\
&= -c_i +  \frac{B}{N-1}(P_H-P_L) E[2p_j-1]~.
\end{align*}
This leads to the same equilibrium equation, which will further lead to the same set of threshold equilibrium analysis.
%where the expectation in last equation is w.r.t. the randomness of $p_i$s.

With (GA), for each specified threshold $c$, a worker's utility difference (under approximation) can be written as 
\begin{align*}
\partial u_i = E[1-2[(\alpha^i-1)F(c)+1]^{N-1}]~,
\end{align*}
where the expectation is over the randomness in $\alpha^i := \text{exp}(-2(P^i_H-P^i_L)^2)$. We first can show:
\begin{lemma}
$e^{-2x^2}$ is concave on $x \in [0,\frac{1}{2}]$~.\label{lemma:e2x}
\end{lemma}
%This is fairly easy to prove by taking the second order derivatives. 

Now notice $g(x):=[1-2[(x-1)F(c)+1]^{N-1}]$ is concave in $x$ for $x \geq 0$ and $N \geq 2$, and we know $1-2[(\alpha^i-1)F(c)+1]^{N-1}$ is also concave in both $P^i_H,P^i_L$ (since $P^i_H \geq P^i_L \geq 0.5$ we have $0 \leq P^i_H-P^i_L \leq 0.5$, and then we can apply Lemma \ref{lemma:e2x}), by composition theory. Therefore we have
\begin{align*}
&E[1-2[(\alpha^i-1)F(c)+1]^{N-1}] \\
&\geq 1-2[(e^{-2(E[P^i_H-P^i_L])^2}-1)F(c)+1]^{N-1} \\
&= 1-2[(e^{-2(P_H-P_L)^2}-1)F(c)+1]^{N-1}~.
\end{align*}
Using above relaxation the rest of analysis will again follow.

\end{proof}

%\section{Proof for Lemma \ref{truthtelling}}

%\section{Proof for Lemma \ref{eqn:mv}}
%
%\begin{proof}
%
%\end{proof}
\section{Proof for Theorem \ref{eqn:mv}}

\begin{proof}
%
%For simplicity we will also write $L_i \in \{1(\sum_{j \in \mathcal U, j \neq i} L_j \geq 0.5)\}$ as $L_i = 1(\frac{\sum_{j \in \mathcal U, j \neq i} L_j}{N-1})$. 
%First we have
%$
% E[S(L_i, \{L_{j}\}_{-i},B)] = BP^B_i(\{e_j\}_{j \in \mathcal U}, B),~$
%where $ P^B_i(\{e_j\}_{j \in \mathcal U}, B)$ is the probability worker $i$ matches the majority answer. Instead of messing with the tie-breaking rule, we grant a bonus whenever
%$
%L_i \in \{1(\sum_{j \in \mathcal U, j \neq i} L_j \geq 0.5)\}.% |\{e_j\}_{j \in \mathcal U}, B )]~.
%$
For simplicity we will also write $L_i \in L_M$ as $L_i = L_M$. Further denote $I_i$ as the indicator variable for whether worker $i$ has correctly labeled the instance, clearly
$P(I_i = 1) = p_i(e_i,r_i)$. 
With this we have
\begin{align}
&P(L_i = L_M |\{e_j\}_{j \in \mathcal U}, B ) \\
& =p_i P(\frac{\sum_{j \in \mathcal U, j \neq i} I_j}{N-1}\geq 0.5)|\{e_j\}_{j \neq i}, B )+(1-p_i) P(\frac{\sum_{j \in \mathcal U, j \neq i} I_j}{N-1}< 0.5|\{e_j\}_{j \neq i}, B )\nonumber \\
&=p_i \biggl(2P(\frac{\sum_{j \in \mathcal U, j \neq i} I_j}{N-1} \geq 0.5|\{e_j\}_{j \neq i}, B )-1\biggr)+ P(\frac{\sum_{j \in \mathcal U, j \neq i} I_j}{N-1}< 0.5|\{e_j\}_{j \neq i}, B )~.\label{pb}
\end{align}
Shorthand the majority voting term $P(\frac{\sum_{j \in \mathcal U, j \neq i} I_j}{N-1} \geq 0.5|\{e_j\}_{j \neq i}, B )$ as $P(B,c^*)$. $u_i(e_i = 1,r_i=1) - u_i(e_i = 0,r_i=1)$ can then be written as
$
%&)\\
%&= -c_i + B \cdot \biggl(P^B_i(\{e_j\}_{j \neq i}, b, B|e_i=1) - P^B_i(\{e_j\}_{j \neq i}, b, B|e_i=0) \biggr) \\
-c_i + B(P_H-P_L) (2P(B,c^*) - 1 ).~
%-c_i + B(P_H-P_L) \biggl(2P(\frac{\sum_{j \in \mathcal U, j \neq i} I_j}{N-1}\geq 0.5)|\{e_j\}_{j \neq i}, B ) - 1 \biggr)~.
$
Similarly to (PA), by setting above term to 0 we get the equilibrium equation. 

Again consider the threshold policy, we can further express the majority voting term $P(B,c^*)$ in details:
\begin{align*}
&P(\frac{\sum_{j \in \mathcal U, j\neq i} I_j}{N-1} \geq 0.5)|\{e_j\}_{j \neq i}, b, B )\\
= \sum_{k=0}^{N-1} &{N-1 \choose k} F^{k}(c^*)(1-F(c^*))^{N-1-k} \cdot P(\frac{\sum_{j \in [k]} I^H_j + \sum_{j \in \mathcal U_{-i}-[k]} I^L_k}{N-1} \geq 0.5)~,
\end{align*}
where we have used $[k]$ to denote the set of users who exerted high efforts. Notice by Chernoff bound we have for each term in the summation above(when $k \geq 1$)
\begin{align*}
P&(\frac{\sum_{j \in [k]} I^H_j + \sum_{j \in \mathcal U_{-i}-[k]} I^L_k}{N-1} \geq 0.5)\\
& \geq 1-\text{exp}(-2(E[\frac{\sum_{j \in [k]} I^H_j + \sum_{j \in \mathcal U_{-i}-[k]} I^L_k}{N-1}]-0.5)^2 (N-1))\\
&= 1-\text{exp}(-2(\frac{k(P_H-P_L)+(N-1)P_L}{N-1}-0.5)^2 (N-1))\\
&\approx 1 - \text{exp}(-2\frac{(P_H-P_L)^2}{N-1} k^2) \\
&\geq 1 - \text{exp}(-2(P_H-P_L)^2 k)~.
\end{align*}
The approximation is due to the fact when $P_L$ is slightly larger than 0.5, we have
$
\frac{(N-1)P_L}{N-1} - 0.5 \approx 0.~
$
Plug back%We know $0< \alpha < 1$. Then
\begin{align*}
&\sum_{k=0}^{N-1} {N-1 \choose k} F^{k}(c^*)(1-F(c^*))^{N-1-k} \cdot (1-\alpha^{k})\\
&=1-\sum_{k=0}^{N-1} {N-1 \choose k} (\alpha F(c^*))^{k}(1-F(c^*))^{N-1-k} \\
&=1-[(\alpha-1)F(c^*)+1]^{N-1}~.
\end{align*}
Plug the approximation into the equilibrium equation we proved the equilibrium equation (\ref{equ:eq:ga}) for (GA).

~\\

%
%\begin{align*}
%&2\sum_{k=0}^{N-1} {N-1 \choose k} F^{k}(c^*)(1-F(c^*))^{N-1-k} \cdot (1-\alpha^{k})-1 = \frac{c^*}{B(P_H-P_L)}~.%\label{equilibrium}
%\end{align*}
Most of the rest equilibrium analysis is similar with the case for PA so we omit the details. The only difference is as follows: notice $[(\alpha-1)y+1]^{N-1}$ is a non-increasing convex function in $y$, so we have $[(\alpha-1)F(c^*)+1]^{N-1}$ is convex by composition theory. We thus establishes the concavity of $1-2[(\alpha-1)F(c^*)+1]^{N-1}$ for $N \geq 2$. Then  if 
$$
1-\alpha^{N-1} \geq \frac{c_{\max}}{B(P_H-P_L)}~,
$$
we will have
\begin{align*}
1-2[(\alpha-1)F(c^*)+1]^{N-1} > \frac{c^*}{B(P_H-P_L)}, \forall c^* < c_{\max}~.
\end{align*}
Otherwise setting above inequality to equality returns the only equilibria. Notice in our approximation, we have used strictly greater than or equal to (except for the approximation $\frac{(N-1)P_L}{N-1} - 0.5 \approx 0$ which will not affect the conclusion above), the above solution thus provides a lower bound on the bonus level (for fixed threshold) or an upper bound on the effort level (for fixed bonus) for the original one.
\end{proof}

\section{Proof for Lemma \ref{comp}}

\begin{proof}
Again via Chernoff bound we know, 
\begin{align*} 
&P(\frac{\sum_{j \in [k]} I^H_j + \sum_{j \in \mathcal U_{-i}-[k]} I^L_k}{N-1} \geq 0.5) \\
&\geq 1-\text{exp}(-2(E[\frac{\sum_{j \in [k]} I^H_j + \sum_{j \in \mathcal U_{-i}-[k]} I^L_k}{N-1}]-0.5)^2 (N-1))\\
&= 1-\text{exp}(-2(\frac{k(P_H-P_L)+(N-1)P_L}{N-1}-0.5)^2 (N-1))~.
\end{align*}
Notice
$$
\frac{k(P_H-P_L)+(N-1)P_L}{N-1}-0.5 > P_L - 0.5 > 0,~
$$ thus when $N$ is large, $P(\frac{\sum_{j \in [k]} I^H_j + \sum_{j \in \mathcal U_{-i}-[k]} I^L_k}{N-1} \geq 0.5)$ can be made arbitrarily close to 1, so that it is larger than $E[\frac{\sum_{j \in [k]} I^H_j + \sum_{j \in \mathcal U_{-i}-[k]} I^L_k}{N-1}]$. Specifically a sufficient condition is given by
\begin{align*}
1-\text{exp}&(-2(P_L - 0.5)^2 (N-1)) > P_H \geq E[\frac{\sum_{j \in [k]} I^H_j + \sum_{j \in \mathcal U_{-i}-[k]} I^L_k}{N-1}]\\
&\Rightarrow N > \frac{-\log (1-P_H)}{2(P_L-0.5)^2}+1
\end{align*}

Therefore
\begin{align*}
\sum_{k=0}^{N-1}& {N-1 \choose k} F^{k}(c^*)(1-F(c^*))^{N-1-k} \cdot P(\frac{\sum_{j \in [k]} I^H_j + \sum_{j \in \mathcal U_{-i}-[k]} I^L_k}{N-1} \geq 0.5) \\
&> \sum_{k=0}^{N-1} {N-1 \choose k} F^{k}(c^*)(1-F(c^*))^{N-1-k} E[\frac{\sum_{j \in [k]} I^H_j + \sum_{j \in \mathcal U_{-i}-[k]} I^L_k}{N-1}]~.%\\
%&\geq \frac{c^*}{B(P_H-P_L)}~.
\end{align*}
Therefore as $B$ is an inverse function of above probability we will have $B_{\text{PA}} > B_{\text{GA}}$.

%Now we calculate the total cost for bonus. From the equilibrium equation we know
%\begin{align*}
%& P(B, c^*) -1 = \frac{c^*}{B(P_H-P_L)} \Rightarrow B\cdot P(B, c^*) = \frac{c^*}{P_H-P_L}+\frac{B}{2}~.
%\end{align*}
%Since $B_{\text{GA}} < B_{\text{PA}}$ and as above applies to all workers, we know group output agreement is a more cost-efficient mechanism.
%Given each $B$, there exists an equilibrium point $c^*$. Based on $c^*$, and assuming knowing $F(c)$, 
Now we calculate the total cost for bonus.
\begin{align*}
&B \cdot E[N_e(B)] = N \cdot B P(B,c^*)=B(F(c^*)T_1 + (1-F(c^*))T_2) = B F(c^*)(T_1-T_2) + B T_2~,
%= N(\frac{c^*}{P_H-P_L}+\frac{B}{2})\\
%&= \frac{Nc^*}{P_H-P_L} + \frac{Nc^*}{4(P_H-P_L)F(c^*)}~.%N\cdot \sum_{k =0}^{N-1} &{N-1 \choose k} F^{k}(c^*(B))(1-F(c^*(B)))^{N-1-k}\cdot P(\frac{\sum_{j \in [k]} I^H_j + \sum_{k \in \mathcal U_{-i}-[k]} I^L_k}{N-1} \geq 0.5)~,
\end{align*}
where $T_1$ is the probability of getting bonus when a worker exerts effort, and $T_2$ is the one when no effort is exerted. Through equilibrium equation, under both (PA) and (GA) we know 
$
T_1-T_2 = \frac{c^*}{B}
$
 as:
 \begin{align}
 -c^* + BT_1 - BT_2 = 0~.
 \end{align}
Therefore 
\begin{align*}
B F(c^*)(T_1-T_2) + B T_2 = c^* F(c^*) + BT_2~.
\end{align*}
For $T_2$ we have
\begin{align*}
T_2 = (2P_L-1)P(\{L_j\}_{j \neq i}) + (1-P_L)~,
\end{align*}
where $P(\{L_j\}_{j \neq i})$ is the probability the answer worker $i$ is matching is correct. Notice we also have
\begin{align*}
B = \frac{c}{(P_H-P_L)(2P(\{L_j\}_{j \neq i})-1)},
\end{align*}
we then have
\begin{align*}
BT_2 = \frac{c(2P_L-1)P(\{L_j\}_{j \neq i}) }{(P_H-P_L)(2P(\{L_j\}_{j \neq i})-1)}+B(1-P_L)~.
\end{align*}
First as above we can prove $P_{\text{PA}}(\{L_j\}_{j \neq i}) < P_{\text{GA}}(\{L_j\}_{j \neq i})$, and as $\frac{c(2P_L-1)P(\{L_j\}_{j \neq i}) }{(P_H-P_L)(2P(\{L_j\}_{j \neq i})-1)}$ is a decreasing function in $P(\{L_j\}_{j \neq i})$ we know
\begin{align*}
 \frac{c(2P_L-1)P_{\text{(PA)}}(\{L_j\}_{j \neq i}) }{(P_H-P_L)(2P_{\text{(PA)}}(\{L_j\}_{j \neq i})-1)} >  \frac{c(2P_L-1)P_{\text{(GA)}}(\{L_j\}_{j \neq i}) }{(P_H-P_L)(2P_{\text{(GA)}}(\{L_j\}_{j \neq i})-1)}. 
\end{align*}
Combined with the fact $B_{\text{PA}} > B_{\text{GA}}$ we finish the proof.
\end{proof}

\section{Proof for Proposition \ref{convex}}

\begin{proof}
First as shown in previous proof we have
\begin{align*}
B \cdot E[N_e(B)] &= c^* F(c^*) + \frac{c^*(2P_L-1)P(\{L_j\}_{j \neq i}) }{(P_H-P_L)(2P(\{L_j\}_{j \neq i})-1)}+B(1-P_L)\\
&=c^* F(c^*) + c^* \frac{2P_L-1}{2(P_H-P_L)} + \frac{1}{P(\{L_j\}_{j \neq i})-1/2}+B(1-P_L)
\end{align*}
Under (PA) we have 
\begin{align*}
P(\{L_j\}_{j \neq i}) = F(c^*)P_H + (1-F(c^*)) P_L = (P_H-P_L) F(c^*) + P_L~,
\end{align*}
which is concave in $c^*$, and then by composition theory we know the third term is convex, as $1/(x+a)$ is convex for $x \geq 0$ when $a \geq 0$. 

Consider the last term $B$. Since 
\begin{align*}
B = \frac{c^*}{(P_H-P_L)(2(P_H-P_L)F(c^*)+2P_L-1)} = \frac{1}{2(P_H-P_L)^2} \cdot \frac{c^*}{F(c^*)+a}~,
\end{align*}
where 
$
a := \frac{2P_L-1}{P_H-P_L} > 0~.
$
When $P_L \approx 0.5$, $2P_L-1 \approx 0$. When we omit $a$, we only need to prove $\frac{c}{F(c)}$ is convex. This approximation is not entirely unreasonable when our targetted $c$ is bounded away from $0$. We have the following results:
%Plug in the second term we have
%\begin{align*}
% \frac{c^*(2P_L-1) ((P_H-P_L) F(c^*) + P_L)}{(P_H-P_L)(2P(\{L_j\}_{j \neq i})-1)}
%\end{align*}

%\com{to be added}
%Since $1-[(\alpha-1)F(c^*)+1]^{N},  -\frac{Nc^*}{P_H-P_L}$ are concave in $c^*$, when $\frac{c}{F(c)}$ is convex, the above optimization problem is concave, whose optimal solution can be calculated efficiently (in polynomial time). 
%
%

\begin{proposition}
When $f(c)$ is twice differentiable and $\frac{\partial^2 f(c)}{\partial^2 c} \geq 0$, $\frac{c}{F(c)}$ is convex.
\end{proposition}
\begin{proof}
To see this, first study the convexity of $F(c)/c$ instead. 
%\begin{align*}
%\frac{F(c)+a}{c} = \frac{F(c)}{c} + \frac{a}{c}~.
%\end{align*}
%Since $a/c$ is convex, we only need to prove $ F(c)/c$ is convex.
Take second order derivatives we have
\begin{align*}
\frac{\partial^2 F(c)/c}{\partial^2 c}= \frac{c2f'(c)-2cf(c)+2F(c)}{c^4}~. 
\end{align*}
Consider the numerator. First
$$
c2f'(c)-2cf(c)+2F(c)|_{c=0} = 0,~
$$
and further taking derivative of above term gives us
$$
(c2f'(c)-2cf(c)+2F(c))' = -c^2\frac{\partial^2 f(c)}{\partial^2 c} \leq 0
$$
by assumption. Then 
$
c2f'(c)-2cf(c)+2F(c) \leq 0, \forall c \geq 0.
$
So $F(c)/c$ is concave. Since $1/x$ is non-increasing convex on $x > 0$, by composition property we know $\frac{1}{F(c)/c}=\frac{c}{F(c)}$ is convex. Since $1-[(\alpha-1)F(c^*)+1]^{N}$ and $-BE[N_e(B)]$ are both concave, we finish the proof.

\end{proof}

Similarly for GA, we only need to prove $B$ is convex in $c$. Since we can write $B$ as $B = \frac{c}{2G(c)-1}$, we first prove $\frac{2G(c)-1}{c}$ is concave, which is equivalent with proving $\frac{G(c)}{c}$ is concave. Take second order derivative we will have
\begin{align*}
\frac{\partial^2 G(c)/c}{\partial^2 c} = \frac{G''c^2+G'c-2G}{c^3}~.
\end{align*}
Again $G''c^2+G'c-2G|_{c=0} < 0$. Take derivative of $G''c^2+G'c-2G$ we have
\begin{align*}
\frac{\partial G''c^2+G'c-2G}{\partial c}
 = \partial^3 G/\partial^3 c c^2 + G'' 3c - G'~.
\end{align*}
Since $G$ is concave and increasing in $c$ we have $ G'' 3c - G'\leq 0$. Thus when $\partial^3 G/\partial^3 \leq 0$ we will be able to establish the concavity of $G(c)/c$, and thus the objective function. 
\end{proof}

\section{Facts that are needed for proving results in Section \ref{sec:data}}

We prove several results that will be repeatedly used. %We first prove the following Lipschitz conditions.
\begin{lemma}
$F(\cdot)$ is Lipschitz with parameter $f(0)$. \label{lemma:f}
\end{lemma}
\begin{proof}
By concavity we know
$
F(c+\delta) - F(c) \leq \max f(c) \delta = f(0) \delta~.
$ 
\end{proof} 

\begin{lemma}
Denote 
$
g(c):=\frac{c}{2(P_H-P_L)F(c)+2P_L-1}~,
$
and we have its inverse function $g^{-1}(\cdot)$ being Lipschitz.\label{lip:g}
\end{lemma}

\begin{lemma}
Over(under)-reporting will lead to higher(lower) estimation of the bonus strategy.\label{bias:B}
\end{lemma}
~\\

These results will be mainly used repeatedly to transfer the bias in reported data into the bias of inferred bonus $B$, and threshold $c$.

\section{Proof for Lemma \ref{lip:g}}

\begin{proof}
First $g(c)$ is strictly increasing in $c$ by noticing:
\begin{align*}
g'(c) = \frac{2(P_H-P_L)F(c)+2P_L-1-c2(P_H-P_L)f(c)}{(2(P_H-P_L)F(c)+2P_L-1)^2} > 0~,
\end{align*}
where the inequality is due to the following fact:
\begin{align*}
&(2(P_H-P_L)F(c)+2P_L-1-c2(P_H-P_L)f(c))' \\
&= 2(P_H-P_L)f(c) - (P_H-P_L)f(c) -c(P_H-P_L) f'(c) > 0,
\end{align*}
that is $(P_H-P_L)F(c)+2P_L-1-cf(c)$ is strictly increasing. So
\begin{align*}
&2(P_H-P_L)F(c)+2P_L-1-cf(c) \\
&\geq 2(P_H-P_L)F(0)+2P_L-1-0f(0) > 0~.
\end{align*}
This is also indicating a larger threshold (better induced effort) corresponds to a higher bonus level. And $g'(c)$ is bounded:
 \begin{align*}
 \min &\{\frac{2P_L-1-2c_{\max}(P_H-P_L)f(0)}{(2P_L-1)^2}, \frac{2P_L-1-2c_{\max}(P_H-P_L)f(0)}{4P^2_L}\} \\
 & \leq g'(c) \leq \frac{2P_H-1+c_{\max} (P_H-P_L) f(0)}{(2P_L-1)^2}~.
 \end{align*}
And further $g(c)$ is everywhere differentiable on $[0,c_{\max}]$. Then $g^{-1}$ is also monotone, differentiable and bounded. So $g^{-1}$ is Lipschitz.
\end{proof}

\section{Proof for Lemma \ref{bias:B}}
\begin{proof}
Suppose there is over-reporting, for each selected threshold $c$ we have
\begin{align*}
\tilde{B}(c) = \frac{c}{(P_H-P_L)(2(P_H-P_L)\tilde{F}(c)+2P_L-1)}~,
\end{align*}
where $\tilde{F}(c)$ is the estimated CDF for $c$ with over-reporting. Notice with over-reporting
\begin{align*}
\tilde{F}(c) \leq F(c),
\end{align*}
and thus $\tilde{B}(c) \geq B(c)$. Similarly we can prove the claim for under reporting.
\end{proof}

\section{Setting up $\delta(t)$ for Theorem \ref{thm:report}}
$\delta(t)$ is mainly chosen so that we do not under-estimate $\tilde{B}$, such that that for workers who reported $\tilde{c}_i(t) \leq c^*(t)$ will be willing to exert effort. For different equilibriums, invoking \emph{over-report} and \emph{under-report} respectively, $\delta(t)$s are set in slightly different ways.

~\\
\emph{Over-reporting}

For equilibrium invoking under-reporting, $\delta(t)$ is mainly used to remove the sampling bias, as: (1) the contributed data are already biased towards calculating a high $B$ (Lemma \ref{bias:B}); (2) for workers who reported less than $c^*(t)$, they have true cost that is also less than $c^*(t)$, so we only need to ensure a bonus level that is no less than $B(c^*(t))$. The only under-estimation comes from sampling uncertainty. We bound it in the following way.
 
By Chernoff bound we have for each $c \in [0,c_{\max}]$,
\begin{align*}
P(|\frac{\sum_{k=1}^{(N-1)t} 1(c(k) \leq c)}{(N-1)t}-F(c)| \geq \frac{1}{\sqrt{(N-1)t/\log t}}) \leq \frac{2}{t^2}~,
\end{align*}
and particularly
$$
P(\frac{\sum_{k=1}^{(N-1)t} 1(c(k) \leq c)}{(N-1)t}-F(c) \leq -\frac{1}{\sqrt{(N-1)t/\log t}}) \leq \frac{1}{t^2}.~
$$
Denote above small perturbation as $\epsilon(t) := 1/\sqrt{(N-1)t/\log t}$. Notice this small perturbation is parameter free, i.e., it is only a function of time $t$, and thus can be calculated on the data requester's side. Now we want to relate the error in estimating $F(\cdot)$ to the bonus level $B$. Notice
\begin{align*}
\tilde{B}_j(c^*(t)) &=  \frac{c^*(t)}{P_H-P_L}\cdot \frac{1}{2(P_H-P_L)\tilde{F}(c^*(t))+2P_L-1} \\
&= \frac{c^*(t)}{P_H-P_L}\cdot \frac{1}{2(P_H-P_L)F(c^*(t))+2(P_H-P_L)\epsilon+2P_L-1}\\
&\geq  \frac{c^*(t)}{P_H-P_L}\cdot (\frac{1}{2(P_H-P_L)F(c^*(t))+2P_L-1}+\frac{(P_H-P_L)}{(2P_H-1)^2} \epsilon)\\
&=B(c^*(t)) + \frac{c^*(t) \epsilon(t)}{(P_H-P_L)^2(2P_H-1)^2}~,
\end{align*}
where the inequality is due to convexity of function $y(x) = \frac{1}{2x+C}, C>0$. So with error probability being at least $1-\frac{1}{t^2}$ we will be having
\begin{align*}
\tilde{B}_j(c^*(t)) \geq B(c^*(t)) -   \frac{c^*(t) \epsilon(t)}{(P_H-P_L)^2(2P_H-1)^2}~.
\end{align*}
So simply set $\delta(t) =  \frac{c^*(t) \epsilon(t)}{(P_H-P_L)^2(2P_H-1)^2}$, we will have
$
\tilde{B}(c^*(t)) + \delta(t) \geq B(c^*(t))~.
$

~\\
\emph{Under-reporting}

For the under-reporting case, $\delta(t)$ now consists of three terms 
\begin{align*}
\delta(t) = \delta_1(t) + \delta_2(t) + \delta_3(t)~.
\end{align*}
As in the over-reporting case, $\delta_1(t)$ is the one for compensating sample bias. $\delta_2(t)$ is for compensating bias in calculating $B$ due to under-reporting. $\delta_3(t)$ is to make sure for workers who under-reported that $\tilde{c}_i(t) \leq c^*(t)$, but having cost $c_i(t) > c^*(t)$ will also be willing to exert efforts.

$\delta_1(t)$ is defined as the sampling bias, which is the same as in the over-reporting case. Now define $\delta_2(t)$. At time $t$, for the selected threshold $c^*(t)$, the unbiased estimation of its CDF is given by
\begin{align*}
\tilde{F}^{\text{unbiased}}(c^*(t)) = \sum_{n=1}^t \sum_{j \neq i} 1(\tilde{c}_j(n) \leq c^*(t)-\epsilon_1(n))~.
\end{align*}
Notice 
$
\tilde{F}^{\text{unbiased}}(c^*(t)) > \tilde{F}(c^*(t))~,
$
Then define $\delta_1(t)$ as follows:
\begin{align*}
\delta_1(t) = \tilde{B}(\tilde{F}^{\text{unbiased}}(c^*(t)) )-\tilde{B}(\tilde{F}(c^*(t)))~.
\end{align*}

For $\delta_3(t)$, notice
\begin{align*}
&B(c^*(t)+\epsilon_1(t))-B(c^*(t)) = \frac{c^*(t)+\epsilon_1(t)}{(P_H-P_L)(2(P_H-P_L)F(c^*(t)+\epsilon_1(t))+2P_L-1)}\\
&~~~~~~~~~- \frac{c^*(t)}{(P_H-P_L)(2(P_H-P_L)F(c^*(t))+2P_L-1)}\\
&\leq \frac{c^*(t)+\epsilon_1(t)}{(P_H-P_L)(2(P_H-P_L)F(c^*(t))+2P_L-1)}\\
&~~~~~~~~~- \frac{c^*(t)}{(P_H-P_L)(2(P_H-P_L)F(c^*(t))+2P_L-1)}\\
&\leq \frac{\epsilon_1(t)}{(P_H-P_L) (2P_L-1)}~. 
\end{align*}
Set $\delta_3(t) = \frac{\epsilon_1(t)}{(P_H-P_L) (2P_L-1)}$. With above we know
$
\tilde{B}_i(t) \geq B(c^*(t))+\epsilon_1(t),~%(\text{over-reporting})~\text{or}~  B(c^*(t)+\epsilon_1(t))(\text{under-reporting}),
$
and then for any worker $i$ who has reported $\tilde{c}_i(t) \leq c^*(t)$, it will be profitable to exert effort. To summarize and help the reader digest, we have the following lemma for each time step $t$.
\begin{lemma}
At each step $t$, after reporting their cost $\{\tilde{c}_i(t)\}$, workers who reported less than or equal to $c^*(t)$ will exert effort and report truthfully, while those who reported higher than $c^*(t)$ will not exert effort, and report truthfully at a $\epsilon$-BNE for (\texttt{M\_Crowd}).
\end{lemma}

\section{Proof for Theorem \ref{thm:report}: bounding over-report} \label{proof:overreport}

\begin{proof}
Suppose a worker deviate from truthfully reporting to $\tilde{c}_i(t) > c_i(t)$ and denote $\sigma_i(t) := \tilde{c}_i(t) - c_i(t) > 0$. We separate the discussion based on the impact of deviation on the \emph{loss} and \emph{profit}. We separate the utility change by under-reporting in two parts: \emph{loss} and \emph{profit}. Denote them as $U_l$ and $U_p$ respectively:
\begin{align*}
\sum_{t'=t}^{T} E[\max_{e_i,r_i} u^{t'}_i(\tilde{c}'_i(t), \mathbf{\tilde{c}_{i,-t}},\mathbf{\tilde{c}_{-i}})]
-E[\max_{e_i,r_i} u^{t'}_i(c_i(t),  \mathbf{\tilde{c}_{i,-t}}, \mathbf{\tilde{c}_{-i}})] = E[U_p]-E[U_l]~.
\end{align*}
~\\
\emph{Lower bound the loss:}

Deviating to a higher report will result in a $\frac{\tilde{c}_i(t)-c_i(t)}{c_{\max}}$ probability of losing the chance of receiving a higher bonus for the current stage, as we randomly select a threshold and exclude the workers who reported higher than the threshold from exerting efforts. Then this part of loss can be bounded below as follows:
\begin{align*}
%&E[u^t_i(\mathbf{\tilde{c}_i}, \mathbf{\tilde{c}_{-i}},e_i=0)] -E[u^t_i(\mathbf{\tilde{c}'_i}, \mathbf{\tilde{c}_{-i}},e_i=1)]
%\\
E[U_l] &\geq \int_{c_i(t)}^{c_i(t)+\sigma_i(t)} ((\tilde{P}(c^*(t),e_i=1)-\tilde{P}(c^*(t),e_i=0))\tilde{B}_i(t)-c_i(t)) \frac{1}{c_{\max}} dc^*(t) \\
 &\geq
 \int_{c_i(t)}^{c_i(t)+\sigma_i(t)} ((P(c^*(t)-\epsilon_2(t)),e_i=1)-P(c^*(t)-\epsilon_2(t)),e_i=0))B(c^*(t))-c_i(t)) \frac{1}{c_{\max}} dc^*(t) \\
 &=  \int_{c_i(t)}^{c_i(t)+\sigma_i(t)} \frac{c^*(t)-\epsilon_2(t)-c_i(t)}{c_{\max}} dc^*(t)\\ 
 &\sigma = c^*(t)-c_i(t) \Rightarrow  := \int_{0}^{\sigma_i(t)} \frac{\sigma-\epsilon_2(t)}{c_{\max}} d\sigma \\
 & = \frac{\sigma^2_i(t)}{2c_{\max}} - \sigma_i(t)\epsilon_2(t)~. %-O(\sigma_i(t)\frac{1}{t^{2-\gamma}})~,
%&=\frac{ P_{\min} }{c_{\max}(P_H-P_L)}\int_{0}^{\delta(t)} \frac{c}{2F(c)+2P_L-1}  dc \\
%&\geq \frac{ P_{\min} }{c_{\max}(P_H-P_L)}\int_{0}^{\delta(t)} \frac{c}{1+2P_L}  dc \\
%& = \frac{ P_{\min} }{c_{\max}(P_H-P_L)(1+2P_L)}\cdot \frac{\delta^2(t)}{2}~.
\end{align*}
where in above equation $\tilde{P}(c^*(t))$ is the probability that worker $i$ will match other's output with threshold $c^*(t)$ and bonus level $\tilde{B}(c^*(t))$. The first inequality is lower bounding the loss by setting $e_i=1$ (which is potentially a sub-optimal solution). The second inequality is due to the fact that $\tilde{B}_i(t) \geq B(c^*(t)$ when $\delta(t)$, the added perturbation, is large enough. Also clearly it is also true that $\tilde{B}_i(t) \geq B(c^*(t)-\epsilon_2(t))$, and the design of our algorithm that only the ones who reported $\tilde{c}_i(t) \leq c^*(t)$, so that $c^*(t) \leq c^*(t)-\epsilon_2(t)$ will be willing to exert effort, $P(c^*(t)-\epsilon_2(t))$ is lower bounding this probability. %The second inequality is due to the equilibrium equation, as well as the fact that workers can choose not to exert effort when $c^*(t)-\epsilon_2(t) < c_i(t)$. 

~\\
\emph{Upper bound the profit:}

For \emph{profit}, over-reporting will change the estimated bonus level, as the data will be utilized for estimating $F(\cdot)$, and further the bonus level $B$. However by design of the algorithm, worker $i$'s data will not affect its own received bonus level $\tilde{B}_i(t)$ as it only depends on the data contributed from all other workers $j \neq i$. But with over-reporting, according to Lemma \ref{bias:B} a higher bonus level will be calculated and offered to other workers, which in turn will induce a higher effort levels from others. This in turn may increase worker $i$'s probability of winning the bonus.

However with probability being at least $1-\frac{1}{t^2}$, further increasing the bonus level will not increase the effort exertion as we have enforced the workers $j$ who has reported $c_j(t) > c^*(t)$ to stay out of the bonus stage, and as since $\tilde{B}(c^*(t)) + \delta(t) \geq B(c^*(t)+\epsilon_1(t))$, all workers who reported $\tilde{c}_j(t) \leq c^*(t)$ are already incentived to exert effort. However with probability at most $\frac{1}{t^2}$, there is indeed a chance that increasing the $\tilde{B}(c^*)$ will further induce more players to participate, due to a bad estimation of $\tilde{B}(c^*(t))$, which could be arbitrarily worse than $B(c^*(t))$. We bound this part of profit. First we prove the following claim:
\begin{lemma}
For (PA), we have
$
0 \leq \tilde{B}_j(c^*(t);\sigma_{i}(t)) - \tilde{B}(c^*(t)) \leq \frac{c^*(t) f(0)}{(P_H-P_L)^2(2P_H-1)^2} \frac{\sigma_i(t)}{t}~.
$
\label{lem:or}
\end{lemma}
%
%\begin{align*}
%\tilde{B}(c^*(t);\sigma_i(t)) &=  \frac{c^*(t)}{P_H-P_L}\cdot \frac{1}{2\tilde{F}(c^*(t))+2P_L-1} \\
%&= \frac{c^*(t)}{P_H-P_L}\cdot \frac{1}{2F(c^*(t)-\sigma_i(t)/t)+2P_L-1}\\
%&\leq \frac{c^*(t)}{P_H-P_L}\cdot \frac{1}{2F(c^*(t))-f(0)\frac{\sigma_i(t)}{t}+2P_L-1}\\
%&\leq  \frac{c^*(t)}{P_H-P_L}\cdot (\frac{1}{2F(c^*(t))+2P_L-1}+\frac{1}{(1+2P_L)^2} f(0)\frac{\sigma_i(t)}{t})\\
%&=B(c^*(t)) + \frac{c^*(t) f(0)}{(P_H-P_L)(1+2P_L)^2} \frac{\sigma_i(t)}{t}~,
%\end{align*}
%where the first inequality is due to the Lipschitz condition of $F(\cdot)$ and the second one is due to convexity. 

The proof is similar to the proof for bounding $\delta(t)$s. The reason we have the perturbation term $\sigma_i(t)/t$ is when the worker deviate by $\sigma_i(t)$, this will only add $\frac{1}{t}$ fraction of mis-reported sample at this stage. Now we want to bound the difference in the corresponding threshold $c$ with (slightly-) different $B$. Again according to the equilibrium equation we know
\begin{align}
2(P_H-P_L)F(c^*(t))+2P_L-1 = \frac{c^*(t)}{B(c^*(t))(P_H-P_L)}~.
\end{align}
Suppose we have two pairs $(\tilde{c},\tilde{B})$ and $(c,B)$, and $\tilde{c} > c$, so $\tilde{B}>B$. We know
\begin{align*}
(P_H-P_L)(\tilde{B}-B) = \frac{\tilde{c}}{2(P_H-P_L)F(\tilde{c})+2P_L-1}-\frac{c}{2(P_H-P_L)F(c)+2P_L-1}.
\end{align*}
Since $g^{-1}(\cdot)$ is Lipschitz,  denote the Lipschitz parameter for $g^{-1}$ as $L$, i.e.,
\begin{align*}
|g^{-1}(x+\delta) - g^{-1}(x)| \leq L \delta, \forall x \in [0,1].
\end{align*}
We then know
\begin{align*}
|\tilde{c}-c| &\leq L|\frac{\tilde{c}}{2(P_H-P_L)F(\tilde{c})+2P_L-1}-\frac{c}{2(P_H-P_L)F(c)+2P_L-1}| \\
&= L(P_H-P_L)(\tilde{B}-B)~.
 \end{align*} 
Now plug in $\tilde{B}:=\tilde{B}_j(c^*(t);\sigma_i(t))$ and $c:=c^*(t)$ we know 
\begin{align*}
|\tilde{c}-c^*(t)| &\leq \frac{c^*(t) f(0) L}{(P_H-P_L)(2P_H-1)^2} \frac{\sigma_i(t)}{t} \\
&\leq \frac{c_{\max} f(0) L}{(P_H-P_L)(2P_H-1)^2} \frac{\sigma_i(t)}{t}~,
\end{align*}
and with this we can bound the additional profit as follows:
\begin{align*}
&\frac{1}{t^2} \cdot  \frac{P(\tilde{c})B(c^*(t))-P(c^*(t))B(c^*(t))}{c_{\max}} = \frac{1}{t^2} \cdot  O(\frac{\tilde{c}-c^*(t)}{c_{\max}})\leq  \frac{c_{\max} f(0) L}{(P_H-P_L)(2P_H-1)^2} O(\frac{\sigma_i(t)}{t^3})~.
\end{align*}

%Suppose $\delta(t) = O(\frac{1}{t^{\alpha}})$, where $\alpha \geq 1$. Then we have
%\begin{align}
%\overline{\delta(t)}^l = (\frac{\sum_{n=1}^t \delta(n)}{t})^l  = (\frac{\sum_{n=1}^t O(\frac{1}{t^{\alpha}})}{t})^l = O(\frac{1}{t^{\alpha l}})
%\end{align}
%Therefore 
%\begin{align*}
%O(\frac{1}{t^{2\alpha}}) \leq O(\frac{1}{t^{2+\alpha l }}).
%\end{align*}
%When $l < 2$, we will have $\alpha \geq \frac{2}{2-l}$, and $\delta(t) \leq \frac{1}{t^{\frac{2}{2-l}}}$. When $l \geq 2$, the above inequality holds for any $\alpha$, in which case, there will be no incentive to deviate at all.

Notice above analysis holds for all $t'\geq t$. Then by a single step deviation we have $E[U_p]$ in the future bounded as follows:
\begin{align*}
E[U_l] \leq \sum_{t'=t}^{\infty} O( \frac{\sigma_i(t)}{(t')^3} ) \leq \frac{\sigma_i(t)}{t^{2-\gamma}}\cdot O(\sum_{t'=t}^{\infty} \frac{1}{(t')^{1+\gamma}}) = O(\frac{\sigma_i(t)}{t^{2-\gamma}})~,
\end{align*}
where the last equality is due to fact the summation series is converging ($1+\gamma>1$). With the loss and profit analysis, then when the following holds:
%Also due to the under-reporting, workers who over-reported to be excluded from effort exertion can enjoy higher probability of matching, due to an increased number of $P_H$ coins. However this fraction of workers will be upper bounded in proportion to $O(\sqrt{\log t/t})$, so will worker $i$'s utility be.
\begin{align*}
 \frac{\sigma^2_i(t)}{2c_{\max}}- \sigma_i(t)\epsilon_2(t) > O(\sigma_i(t)\frac{1}{t^{2-\gamma}})~.
\end{align*}
there will be no incentive for a further deviation. Since $\epsilon_2(t) = O(\frac{1}{t^{2-\gamma}})$ we must have $\sigma_i(t) \leq O(\frac{1}{t^{2-\gamma}})$. %Since $\epsilon_2(t) = O(\frac{1}{t^{2-\gamma}})$, the above condition becomes 

%Therefore we ensured that 
%\begin{align*}
%\frac{\sigma^2_i(t)}{c_{\max}} \leq O(\sigma_i(t)\frac{1}{t^{2-\gamma}}) \Rightarrow \sigma_i(t) \leq ~.
%\end{align*}
%By deviating less than or equal to the order $O(\frac{1}{t^{2-\gamma}})$, we will have the profits is upper bounded by $O(\frac{\sigma_i(t)}{t^{2-\gamma}})=O(\frac{1}{t^{4-2\gamma}})$
%Therefore over time the total profit for deviations at all $T$ steps is bounded by
%\begin{align}
%O(\sum_{t=1}^T \frac{1}{t^{4-2\gamma}}) \leq \text{const.} \Rightarrow \frac{O(\sum_{t=1}^T \frac{1}{t^{4-2\gamma}})}{T} = O(\frac{\text{const.}}{T}) \leq O(\frac{(\log T)^2}{T})~.
%\end{align}
%This finishes the proof.
\end{proof}

%%%%%%%%%%%%%%%
\section{Proof for Theorem \ref{thm:report}: bounding under-report }
\label{proof:under-report}
\begin{proof}
We start with stating the following lemma. The proof is again similar to the one we reported for bounding $\delta$, as well as for Lemma \ref{lem:or} for over-reporting, thus we omit the details.
\begin{lemma}
Suppose worker $i$ is under-reporting by $\sigma_i(t)$ at time $t$, that is 
$
\tilde{c}_i(t) = \max\{c_i(t)-\sigma_i(t), 0\}~,
$
we will then have for (PA):
$
0 \geq \tilde{B}_j(c^*(t); \sigma_i(t)) - B(c^*(t)) \geq - \frac{c^*(t)f(0)}{(P_H-P_L)^2(2P_H-1)^2} \frac{\sigma_i(t)}{t},~\forall j \neq i.
$
\label{map:uncertain}
\end{lemma}

%
%\begin{proof}
%Consider the equilibrium equation for a fixed threshold $c^*$ we have
%\begin{align*}
%B(c^*) =  \frac{c^*}{P_H-P_L}\cdot \frac{1}{2F(c^*)+2P_L-1}~.
%\end{align*}
%Due to the perturbation of the data, the evaluation of $F(c^*)$ will be perturbed, and denote it as $\tilde{F}(c^*)$. Then we have 
%\begin{align*}
%\tilde{B}(c^*;\delta) &=  \frac{c^*}{P_H-P_L}\cdot \frac{1}{2\tilde{F}(c^*)+2P_L-1} \\
%&= \frac{c^*}{P_H-P_L}\cdot \frac{1}{2F(c^*+\delta)+2P_L-1}\\
%&\geq  \frac{c^*}{P_H-P_L}\cdot \frac{1}{2F(c^*)+L\delta^l+2P_L-1}\\
%&\geq  \frac{c^*}{P_H-P_L}\cdot (\frac{1}{2F(c^*)+L\delta^l+2P_L-1}+\frac{1}{(1+2P_L)^2} L\delta^l)\\
%&=B(c^*) + \frac{c^* L}{(P_H-P_L)(1+2P_L)^2} \delta^l~,
%\end{align*}
%where the first inequality is due to the Lipschitz condition for $F(\cdot)$ and the second one is due to convexity of function $y(x) = \frac{1}{2x+C}, C>0$.
%\end{proof}
%

%\begin{assumption}
%Denote the following difference
%\begin{align*}
%d(c^*; \delta) := \tilde{B}(c^*;\delta) + \epsilon(c^*;\delta) - B(c^*).
%\end{align*}
%We assume
%\begin{itemize}
%\item $d(c^*; \delta)$ is independent of $c^*$.
%\item $d(c^*; \delta)$ is on high order of $\delta$, specifically $d(c^*; \delta) = o(\delta^2)$.  
%\end{itemize}
%\end{assumption}
%The first assumption is the Lipschitz and convexity bounds are uniformly tight for all $c^*$. For example, when $F$ is piece-wise linear, the assumption will be satisfied. 

 We again separate the utility change by under-reporting in two parts: \emph{loss} and \emph{profit}. By under reporting it to $c_i(t)-\sigma_i(t)$, we see only when $c^*(t)$ falls into the region $[c_i(t)-\sigma_i(t),c_i(t)]$, worker's utility will change. The expected extra loss comes from the fact that only when worker under-reports and exerts effort, his expected bonus will change. The extra cost incurred for exerting efforts is given as follows (as the sampling for threshold is uniformly random over $[0,c_{\max}]$):
\begin{align*}
E[U_l] = \int_{0}^{\sigma_i(t)} \frac{1}{c_{\max}} \delta ~d\delta= \frac{1}{2c_{\max}}\sigma_i^2(t)~.
\end{align*}
%The integral is uniform over the $\delta(t)$-interval as the sampling for threshold is uniformly random over $[0,c_{\max}]$.

Now consider the profit. Denote the $O(1/t^2)$ Chernoff-type event for the sampling estimation as $\mathcal E(t)$, corresponding to the error term $\epsilon(t)$. We first prove the following results that will relate the uncertainty in $F(\cdot)$ to the uncertainty in $B$: 
%This is again very similar to want we have shown before:
%\begin{align*}
 %\frac{c^*(t) f(0)}{(P_H-P_L)(1+2P_L)^2} \delta_1(t)
%\end{align*}
Again with probability being at least $1-\frac{2}{t^2}$, using Lemma \ref{map:uncertain}, the extra profit is upper bounded as
\begin{align*} 
E[U_p|\mathcal E(t)] P(\mathcal E(t)) \leq \frac{c^*(t) f(0)}{(P_H-P_L)^2(2P_H-1)^2} (\delta_1(t)+\delta_3(t))~.
\end{align*}
With probability being at most $1/t^2$, the bonus level is unpredictable, but suppose it can be upper bounded by $\bar{B}$ (for example the upper limit on the budget for bonus). Then we have
\begin{align}
E[U_p] &= E[U_p|\mathcal E(t)] P(\mathcal E(t)) + E[U_p|\mathcal E(t)] P(\mathcal E(t)) \\
&\leq \sigma_i(t)[\frac{1}{t^2}\cdot \bar{B} + (1-\frac{1}{t^2})(\frac{c^*(t) f(0)}{(P_H-P_L)^2(2P_H-1)^2} (\epsilon(t)+\delta_3(t))~.\label{profit:underreport}%+\underbrace{O(\sigma_i(t))}_{\text{Under-reporting}})]~.
\end{align}
%where the second equality comes from the fact $d(\delta(t))+\epsilon(t)) \leq \bar{B}$, and the Chernoff inequality we presented earlier. 
%In the inequality, due to under-reporting, the matching probability increases by $O(\sqrt{\frac{\log t}{t}})$. 
Therefore there is no incentive for the under-report to go beyond the one satisfying (loss $>$ profit): 
\begin{align}
E[U_l] > E[U_p] &\Leftrightarrow \frac{1}{2c_{\max}}\sigma_i^2(t) >  \sigma_i(t)\biggl [\frac{\bar{B}}{t^2} + (1-\frac{1}{t^2})\frac{c^*(t) f(0)}{(P_H-P_L)^2(2P_H-1)^2} (\epsilon(t)+\delta_3(t))\biggr ]\nonumber\\
&\Rightarrow \sigma_i(t) > \frac{2c_{\max} \bar{B}}{t^2} + 2c_{\max}(\frac{c^*(t) f(0)}{(P_H-P_L)^2(2P_H-1)^2}( \sqrt{\frac{\log t}{(N-1)t}} + O(\sqrt{\frac{\log t}{t}}))~.\label{equilibrium:eqn}
\end{align}
From both above and the argument for over-reporting, we see a profitable deviation satisfies that
\begin{align*}
\sigma_i(t) \leq O(\sqrt{\frac{\log t}{t}}) (\text{as}~O(\frac{1}{t^{2-\gamma}}) < O(\sqrt{\frac{\log t}{t}}))~.
\end{align*}
Short-hand worker $i$'s utility after $t$ as $u^{t'\geq t}_i(\tilde{c}_i(t))$ we have the additional profit is bounded as follows
\begin{align*}
&|E[\max_{e_i,r_i} u^{t'\geq t}_i(\tilde{c}_i(t) = c_i(t)+\sigma_i(t))]-E[\max_{e_i,r_i} u^{t'\geq t}_i(\tilde{c}_i(t) = c_i(t)+\epsilon_{1(2)}(t))]| \\
&\leq |E[\max_{e_i,r_i} u^{t'\geq t}_i(\tilde{c}_i(t) = c_i(t)+\sigma_i(t))]-E[\max_{e_i,r_i} u^{t'\geq t}_i(\tilde{c}_i(t) = c_i(t))]|\\
&~~~+ |E[\max_{e_i,r_i} u^{t'\geq t}_i(\tilde{c}_i(t) = c_i(t))]-E[\max_{e_i,r_i} u^{t'\geq t}_i(\tilde{c}_i(t) = c_i(t)+\epsilon_{1(2)}(t)]|~ \text{(triangle inequality)}  \\
&=O(\sigma_i(t)\cdot \sqrt{\frac{\log t}{t}}) + \epsilon_{1(2)}(t) O(\sqrt{\frac{\log t}{t}}) \leq O(\frac{\log t}{t})~.
\end{align*}
 Therefore over time the total profit for deviations at all $T$ steps is bounded by
\begin{align*}
&|\sum_{t=1}^T E[\max_{e_i,r_i}u^t_i(\mathbf{\tilde{c}_i}, \mathbf{\tilde{c}_{-i}})]/T- \sum_{t=1}^T E[\max_{e_i,r_i} u^t_i(\mathbf{\tilde{c}'_i}, \mathbf{\tilde{c}_{-i}})]/T|\\
&\leq \frac{\sum_{t=1}^T |E[\max_{e_i,r_i} u^{t'\geq t}_i(\tilde{c}_i(t) = c_i(t)+\sigma_i(t))]-E[\max_{e_i,r_i} u^{t'\geq t}_i(\tilde{c}_i(t) = c_i(t)+\epsilon_{1(2)}(t))]|}{T}\\
&\leq \frac{O(\sum_{t=1}^T \frac{\log t}{t})}{T} = \frac{O(\log T \cdot \sum_{t=1}^T \frac{1}{t})}{T} = O(\frac{(\log T)^2}{T})~.
%\\
%& \Rightarrow \frac{O(\sum_{t=1}^T \frac{\log t}{t})}{T} \leq O(\frac{(\log T)^2}{T})~.
\end{align*}
%This finishes the proof.
%To see this, with this setting we will have $
%\delta_1(t) \leq O(\sqrt{\log t/t}),~
%$
%and the order on both sides of Eqn. (\ref{equilibrium:eqn}) will match. 
%
%Otherwise if $\sigma_i(t)$ is set on a much higher order, we will have with high probability
%$
%\tilde{B}_i(t) + \delta_1(t) \leq B(c^{*}(t))~,
%$ under which case the extra gain in bonus will be bounded as
%$
%\sigma_i(t) [\frac{\bar{B}}{t^2}\cdot +(1-\frac{1}{t^2})\epsilon(t)]~,
%$ as no one will under-report. But remember since this is beneficial deviation for each worker we will have
%\begin{align*}
%&\sigma_i(t) [\frac{\bar{B}}{t^2}+(1-\frac{1}{t^2})\epsilon(t)] \geq \frac{\sigma^2_i(t)}{2c_{\max}} \Leftrightarrow \sigma_i(t) \leq O(\epsilon(t))~,
%\end{align*}
%which arrives at a contradiction.

%Moreover the worker wants to maximize $E[U_b-U_c]$. If we can assume all workers are \emph{risk-taking}, that they will substitute the upper bound for $E[U_b]$ for such a calculation, we will have
%\begin{align*}
%\max_{\delta(t)} E[U_b - U_c] \approx  \delta(t)[\frac{2}{t^2}\cdot \bar{B} + (1-\frac{1}{t^2})(d(\delta(t))+\epsilon(t))] - \frac{1}{2c_{\max}}\delta^2(t)
%\end{align*}
%If we further ignore the dependence of $d(\delta(t))$ on $\delta(t)$, the optimal $\delta^*(t)$ can be written in a closed-form:
%\begin{align*}
%\delta^*(t) = c_{\max}[\frac{2}{t^2}\cdot \bar{B} + (1-\frac{1}{t^2})(d(\delta^*(t))+\epsilon(t))]
%\end{align*}
\end{proof}

\section{Proof for Lemma \ref{perf:1}}
\begin{proof}
If all workers follow the $\epsilon$-BNE, we will have 
\begin{align*}
\overline{\sigma(T)} := \frac{\sum_{t=1}^T \sum_i |\tilde{c}_i(t)-c_i(t)|}{NT} \leq O(\frac{N\sum_{t=1}^T \sqrt{\log t/t}}{NT}) \leq O(\sqrt{\frac{\log T}{T}})~,
\end{align*}
where above we have used the fact that $\sum_{t=1}^T 1/\sqrt{t} = O(\sqrt{T})$. Based on Lipschitz condition of $F(\cdot)$ (Lemma \ref{lemma:f}) and Chernoff bound we know we first have the following lemma:
\begin{lemma}
With probability being at least $1-\eta$ we have $\forall c$,
$
|\tilde{F}(c)-F(c)| \leq \sqrt{\frac{\log (2/\eta)}{2NT}}+ f(0) \overline{\sigma(T)}.
$
\label{data:f:map}
\end{lemma}
\begin{proof}
By symmetry, we only need to prove one of the above.Take $|\tilde{F}(\tilde{c}^*)-F(\tilde{c}^*)| $ for example. Notice the following holds:
\begin{align*}
&\frac{\sum_{t=1}^T \sum_i 1(c_i(t) \leq \tilde{c}^*- |\tilde{c}_i(t)-c_i(t)|)}{NT}\leq \\
&~~~~~~~ \tilde{F}(\tilde{c}^*) := \frac{\sum_{t=1}^T \sum_i 1(\tilde{c}_i(t) \leq \tilde{c}^*)}{NT}\\
&\leq \frac{\sum_{t=1}^T \sum_i 1(c_i(t) \leq \tilde{c}^*+ |\tilde{c}_i(t)-c_i(t)|)}{NT} ~,
\end{align*}
where we have used the fact that 
\begin{align*}
- |\tilde{c}_i(t)-c_i(t)| \leq \tilde{c}_i(t) -c_i(t) \leq  |\tilde{c}_i(t)-c_i(t)|~.
\end{align*}
Using Hoeffding bound we know 
\begin{align*}
P(|\frac{\sum_{t=1}^T \sum_i 1(c_i(t) \leq \tilde{c}^* + |\tilde{c}_i(t)-c_i(t)|)}{NT} - \frac{\sum_{t=1}^T \sum_i F( \tilde{c}^*+ |\tilde{c}_i(t)-c_i(t)|)}{NT}| > \sqrt{\frac{\log (2/\eta)}{2NT}}) \leq \eta~,\\ 
P(|\frac{\sum_{t=1}^T \sum_i 1(c_i(t) \leq \tilde{c}^*- |\tilde{c}_i(t)-c_i(t)|)}{NT} - \frac{\sum_{t=1}^T \sum_i F( \tilde{c}^* - |\tilde{c}_i(t)-c_i(t)|)}{NT}| > \sqrt{\frac{\log (2/\eta)}{2NT}}) \leq \eta~.
\end{align*}
Then 
\begin{align*}
\tilde{F}(\tilde{c}^*)-&F(\tilde{c}^*)  = \tilde{F}(\tilde{c}^*)-\frac{\sum_{t=1}^T \sum_i F( \tilde{c}^*+ |\tilde{c}_i(t)-c_i(t)|)}{NT} \\
&+ \frac{\sum_{t=1}^T \sum_i F( \tilde{c}^*+ |\tilde{c}_i(t)-c_i(t)|)}{NT} - F(\tilde{c}^*)\\
&\leq   \sqrt{\frac{\log (2/\eta)}{2NT}} + f(0)\overline{\sigma(T)}~.
\end{align*}
And
\begin{align*}
\tilde{F}(\tilde{c}^*)-&F(\tilde{c}^*)  = \tilde{F}(\tilde{c}^*)-\frac{\sum_{t=1}^T \sum_i F( \tilde{c}^*- |\tilde{c}_i(t)-c_i(t)|)}{NT}\\
& + \frac{\sum_{t=1}^T \sum_i F( \tilde{c}^*- |\tilde{c}_i(t)-c_i(t)|)}{NT} - F(\tilde{c}^*)\\
&\geq  - \sqrt{\frac{\log (2/\eta)}{2NT}} - f(0)\overline{\sigma(T)}~.
\end{align*}
%Proved.
\end{proof}

Using this Lemma we know
\begin{align*}
&|\tilde{F}(\tilde{c}^*)-F(\tilde{c}^*)| \leq \sqrt{\frac{\log (2/\eta)}{2T}}+ f(0) \overline{\sigma(T)},\\
&|\tilde{F}(c^*)-F(c^*)| \leq \sqrt{\frac{\log (2/\eta)}{2T}}+ f(0) \overline{\sigma(T)}.
\end{align*}

Suppose $c^*$ is bounded away from $0$, that is the data requester is targeting non-trivial effort exertion. Then by concavity of $U_D$ (in $F(c)$), and the boundedness of its first order derivative for $c^*>0$, we know we have Lipschitz condition for $U_D$. So
\begin{align*}
|\tilde{U}_D(B^*) - U_D(B^*)| \leq O(  \sqrt{\frac{\log (2/\eta)}{2NT}} + f(0)\overline{\sigma(T)})~.
\end{align*}
Since $\tilde{U}_D(\tilde{B^*}) \geq \tilde{U}_D(B^*)$ (by optimality of $\tilde{B^*}$ for $\tilde{U}_D(\cdot)$) we know
\begin{align*}
U_D(\tilde{B}^*)-U_D(B^*) &\geq U_D(\tilde{B}^*)-\tilde{U}_D(\tilde{B}^*) + \tilde{U}_D(B^*)-U_D(B^*) \\
&\geq -2O(  \sqrt{\frac{\log (2/\eta)}{2NT}}+ f(0) \overline{\sigma(T)})~.
\end{align*}
Take $\eta:=1/T^2$, we achieve a regret bounded on the order of $\sqrt{\log T/T}$.
\end{proof}

%\subsection{Setup the learning problem}
\section{Proof for Theorem \ref{learn2}}

\begin{proof}
First notice the expected number of exploration phases up to any time $t$ can be lower bounded by
\begin{align}
E[\sum_{t'=1}^t &1\{s(t')\}] =\sum_{t'=1}^t  \min \{1, \frac{\log T}{t^{1-z}}\} \\
&\geq O(t^{z}\log t) - o(\log^{\frac{1}{1-z}} T)~,\label{num:exp}
\end{align}
where the $o(\cdot)$ term is to compensate the loss when $\frac{\log T}{t^{1-z}} > 1$. Moreover using Hoeffding's inequality we also know $\forall ~0<\kappa < 1$
\begin{align}
P&(\sum_{t'=1}^t 1\{s(t')\}-E[\sum_{t'=1}^t 1\{s(t')\}]<-\kappa E[\sum_{t'=1}^t 1\{s(t')\}])\leq e^{-\frac{\kappa^2}{2} t^{z}\log t } \leq \frac{1}{t^2}~,\label{num:exp:prob}
\end{align}
when $t \geq (\frac{1}{\kappa^2})^{1/z}$. This result will ensure at any time $t$, w.h.p. the data requester will have enough number of sample points that are elicited from exploration phases.
%The basic idea of the proof is as follows: when a worker decides to deviate, only with probability being no more than $\frac{\log T}{t^{z_1+z_2}}$ her reported data will come into play for calculating the bonus in the future. Notice this is a polynomially decreasing factor. While on the other hand, with probability being more than $1-\frac{\log T}{t^{z_1}}$, the worker will risk missing an exploration bonus (if over-reporting by too much). 

~\\
\emph{Under-reporting}

For under-reporting, as we argued in previous section, this will not increase the bonus level, but will only increase a worker's chance of getting the bonus. An $\epsilon_1(t):=O(\sqrt{\log t/t})$ upper bound was proved in Section \ref{proof:under-report}. However notice with (\texttt{RM\_Crowd}), at time $t$ we do not have $t$ samples for estimating bonus level (because we only use data from exploration phases). Instead we have:
\begin{align*}
\epsilon_1(t) \leq \sqrt{\log n(t)/n(t)},~\text{where}~ n(t) = \sum_{t'=1}^t 1(s(t'))~. 
\end{align*}
Note since $\sqrt{\log x/x}$ is strictly decreasing when $x$ is larger than certain constant ($<10$), w.p. being at least $1-\frac{2}{t^2}$ we have
\begin{align*}
 \sqrt{\frac{\log n(t)}{n(t)}} \leq O(\sqrt{\frac{\log (t^{z} \log t)}{t^{z} \log t}}) = O(\frac{z}{t^{z/2} })~.
\end{align*}
%This finishes the proof.

~\\
\emph{Over-reporting}

Now realizing the data requester may utilize the collected data to calculate the optimal bonus level, workers have stronger motivation to over-report, as a mis-report creates more bias in future bonus calculation as the number of exploration data grows slower. Suppose worker $i$ over-report by $\sigma_i(t)$ at time $t$. Again similar with our analysis in Section \ref{proof:overreport}, the loss with over-reporting is lower bounded by
$
\frac{\log T}{t^{1-z}}\frac{(\sigma_i(t)-\epsilon_2(t))^2}{2c_{\max}}.~
$
 While on the other hand, there is a probability being $\frac{\log T}{t^{1-z}}$, the over-reported data will be included in future bonus calculation. This additional data will change the future calculation by $O(\frac{\sigma_i(t)}{n(t')}), \forall t' \geq t$. Different from the proof in Section \ref{proof:overreport}, due to the dis-continuous sampling, the accumulated profit differ from $O(\sum_{t'=t}\frac{\sigma_i(t)}{t'})$, since the number of data utilized for calculation is much smaller for certain $t$, and grows much slower in $t$. We nevertheless can lower bound $n(t')$ as follows: combine Eqn.(\ref{num:exp}) and (\ref{num:exp:prob}), we know $\forall t'$:
$
P(n(t') \leq O((t')^{z}\log (t'))) \leq \frac{1}{(t')^2}.~
$
So the additional profit for such deviation is bounded by (changing $t$ to $t^{z}\log t$ in Eqn.(\ref{profit:underreport}))
\begin{align*}
&\sigma_i(t)(\frac{1}{( t^{z}\log t)^2}O(\sum_{t'=t} \frac{1}{ (t')^{z}\log (t')}) +O(\sum_{t'=t}\frac{1}{(t')^2})) =\sigma_i(t) \cdot \frac{1}{( t^{z}\log t)^2}O(\sum_{t'=t} \frac{1}{ (t')^{z}\log (t')}) \\
%&\leq \frac{1}{( (t)^{z}\log (t))^2}O(\sum_{t'=t} \frac{\sigma_i(t)}{ (t')^{z}\log (t')}) ~.
\end{align*}
Consider the sum series.
\begin{align*}
\frac{1}{( t^{z}\log t)^2}\sum_{t'=t} \frac{1}{ (t')^{z}\log (t')} \leq \frac{1}{(\log t)^3 t^{3z-1-\gamma}} \sum_{t'=t} \frac{1}{ (t')^{1+\gamma}}=O(\frac{1}{(\log t)^3 t^{3z-1-\gamma}})~,
\end{align*}
where $\gamma>0$ is an arbitrarily small quantity. $3z>1$ is needed for a converging term. This gives us that any profitable deviation satisfies: 
$
\sigma_i(t) \leq O(\frac{1}{(\log t)^3 t^{3z-1-\gamma}}).~ %, O(\frac{z}{t^{z/2} })\} = O(\frac{z}{t^{z/2} }),~\text{when}~ z > 2/5.
$ Also notice the following holds:
\begin{align*}
z > \frac{2(1+\gamma)}{5} \Rightarrow 3z-1-\gamma>z/2 \Rightarrow O(\frac{1}{(\log t)^3 t^{3z-1-\gamma}}) < O(\frac{z}{t^{z/2}})~.
\end{align*}
The rest of analysis is similar to the one for (\texttt{M\_Crowd}), with total profit for deviations at all $T$ steps is bounded by
\begin{align*}
O(\sum_{t=1}^T \frac{z^2}{t^{z}})  \leq O(z^2 T^{1-z})  \Rightarrow \frac{O(\sum_{t=1}^T  \frac{z^2}{t^{z}})}{T} \leq O(\frac{z^2 T^{1-z}}{T}) = O(\frac{z^2}{T^z})~.
\end{align*}
\end{proof}

\section{Proof for Lemma \ref{perf:2}}

\begin{proof}
$R(T)$ mainly consists of two parts: (1) due to \emph{exploration} phases (2) due to imperfect estimation for \emph{exploitation} phases:
\begin{align*}
R(T) = R_{\text{explore}}(T) + R_{\text{exploit}}(T)~.
\end{align*}
The exploration regret is easy to characterize:
\begin{align*}
R_{\text{explore}}(T) \leq [1+(b+\bar{B})N]O(T^z\log T)~,
\end{align*}
as $-(b+\bar{B})N \leq U_D \leq 1$, and the expected number of exploration phases is upper bounded on the order of $T^z\log T$ with high probability.

For $R_{\text{exploit}}(T)$, consider each $t$ that is in exploitation phase. During exploitation phases, workers will not over-report. Due to under-reporting, the induced answer will achieve a better aggregated accuracy $P^c(N,B)$ (as we ensured workers who under-reported would be willing to exert efforts). But we do loss in term of more offered bonus, as in this case the matching probability of two answers will increase. Nevertheless there are only $O(z/t^{z/2})$ fraction of workers under-report, the probability of giving out a bonus also changes at this order. Thus the over-payed bonus is upper bounded by $\sum_{t=1}^T O(1/t^{z/2}) = O(T^{1-z/2})~.
$

%a worker may deviate to over-report in exploitation phases, that is because once the worker is excluded from effort exertion, her probability of winning the bonus will be depending on how many workers reported less than the threshold -- the more the better. However this fraction of workers will be upper bounded $O(1/t^{z/2})$ (proportional to number of under-reported the workers).  So there are at most $O(1/t^{z/2})$ less effort exerted. We can prove $P^c(N,B)$ is concave in the probability of effort exertion, we know this part of regret is bounded by
%\begin{align*}
%\sum_{t=1}^T O(1/t^{z/2}) = O(T^{1-z/2})~.
%\end{align*}

Now we bound the regret due to a noisy calculation of $B$. We know with probability being at least $1-O(\frac{1}{t^2})$ (such that the number of samples collected from exploration phases is at least $O(t^z\log t)$), we have the average sampling bias incurred by mis-reporting bounded as
\begin{align*}
\overline{\sigma(t)} &\leq O(\frac{c_{\max}+\sum_{t'=1}^{t^z\log t} \frac{1}{(t')^{z/2}}}{t^z \log t}) \\
&= O(\frac{1}{t^{z-z^2/2}\log t})~.
\end{align*}
Meanwhile using Chernoff bound we have with probability being at least $1-\frac{2}{t^2}$ we have sampling bias for cost data realization can be bounded as:
$
\epsilon(t) \leq O(\frac{1}{t^{z/2}}),
$
since we have $O(t^z\log t)$ samples. Similar to our previous analysis for Lemma \ref{data:f:map}, such bias in cost data distribution can be mapped to the bias in $F$. Again using the concavity of $U_D$ and the boundedness of its first order derivatives we know the loss in utility function is proportional to the bias in samplings, which is $O(\frac{1}{t^{z/2}}+\frac{1}{t^{z-z^2/2}\log t})$. Then summarize above discussion we have
\begin{align*}
R_{\text{exploit}}(T) &\leq O(T^{1-z/2})+ \sum_{t=1}^T O(\frac{1}{t^{z/2}}+\frac{1}{t^{z-z^2/2}\log t}) +\text{const.} \\
& =O(T^{1-z/2}+T^{1-z+z^2/2})~.
\end{align*}
Since 
$
0 < z \leq 1 \Rightarrow 1-z+z^2/2 \leq 1-z/2.~,
$
we then have the bound reduce to 
$
R_{\text{exploit}}(T) \leq O(T^{1-z/2}).
$
Combine exploration and exploitation analysis we know
$
R(T) \leq O(T^z\log T + T^{1-z/2}),~
$
and the best $z$ is when $1-z/2 = z \Rightarrow z = \frac{2}{3}$, which leads to a bound at the order of $O(T^{2/3}\log T)$.

\end{proof}

\end{document}